\documentclass[11pt,sort&compress]{amsart}
\usepackage[utf8]{inputenc}
\usepackage{amsmath}
\usepackage{array, amsfonts}
\usepackage{amsthm}
\usepackage{amssymb}
\usepackage{enumerate}
\usepackage{lineno}
\usepackage[bookmarks=true]{hyperref}
\usepackage{enumitem}
\usepackage{mathrsfs}
\usepackage{stackengine}
\usepackage{bookmark}
\usepackage{amsrefs}
\usepackage{color,soul}
\usepackage[draft]{todonotes}
\usepackage{fancyhdr}
\usepackage{MnSymbol}
\usepackage{algorithm2e}
\usepackage[noend]{algpseudocode}
\usepackage{mathtools}
\usepackage[foot]{amsaddr}

\usepackage[a4paper, left=3cm, right=3cm, top=3cm, bottom=3cm]{geometry}

\definecolor{darkgreen}{rgb}{0,0.5,0}

\newcommand{\comment}[1]

\newcommand\myfunc[5]{%
	\begingroup
	\setlength\arraycolsep{0pt}
	#1\colon\begin{array}[t]{c >{{}}c<{{}} c}
		#2 & \to & #3 \\ #4 & \mapsto & #5 
	\end{array}%
	\endgroup}

\newcommand{\RNum}[1]{\uppercase\expandafter{\romannumeral #1\relax}}

\title{ On sequential structures in\\ incompressible multidimensional networks\footnote{In~\cite{Abrahao2019bams}, a preliminary version of part of this paper is presented as an extended abstract.} }

\author{Felipe S. Abrah\~{a}o}
\address[Felipe S. Abrah\~{a}o, Klaus Wehmuth, Artur Ziviani]{National Laboratory for Scientific Computing (LNCC), Petr\'{o}polis, Brazil}
\address[Felipe S. Abrah\~{a}o]{Centre for Logic, Epistemology and the History of Science (CLE), University of Campinas (UNICAMP), Brazil}
\address[Felipe S. Abrah\~{a}o, Hector Zenil]{Oxford Immune Algorithmics, Oxford University Innovation, Oxford, U.K.}
\email{fabrahao@unicamp.br}

\author[Klaus Wehmuth]{Klaus Wehmuth} 
\email{klaus@lncc.br}

\author{Hector Zenil}  
\address[Hector Zenil]{Algorithmic Dynamics Lab, School of Biomedical Engineering and Imaging Sciences \& King's Institute for Artificial Intelligence, King's College London, U.K.. 
The Alan Turing Institute, British Library, London, U.K..
}
\email{hector.zenil@kcl.ac.uk}

\author[Artur Ziviani]{Artur Ziviani}
\email{ziviani@lncc.br}

\fancypagestyle{style1}{
	\fancyhead{}
	\fancyhead[L]{}
	\fancyhead[R]{\thepage}
	\fancyfoot{}
}

\fancypagestyle{style2}{
	\fancyhead{}
	\fancyhead[L]{Topological properties of high-order networks}
	\fancyhead[R]{\thepage}
	\fancyfoot{}
}

\fancypagestyle{plain}{
	\fancyhead{}
	\fancyhead[C]{ \huge{ \textbf{LAGOS 2019} } \\ \normalsize{\textit{Extended Abstract} } }
}
\begin{document}

\begin{abstract}\label{abstract}
	In order to deal with multidimensional structure representations of real-world networks, as well as with their worst-case irreducible information content analysis, the demand for new graph abstractions increases. This article investigates incompressible multidimensional networks defined by generalized graph representations. In particular, we mathematically study the lossless incompressibility of snapshot-dynamic networks and multiplex networks in comparison to the lossless incompressibility of more general forms of dynamic networks and multilayer networks, from which snapshot-dynamic networks or multiplex networks are particular cases. Our theoretical investigation first explores fundamental and basic conditions for connecting the sequential growth of information with sequential interdimensional structures such as time in dynamic networks, and secondly it presents open problems demanding future investigation. Although there may be a dissonance between sequential information dynamics and sequential topology in the general case, we demonstrate that incompressibility (or algorithmic randomness) dissolves it, preventing both the algorithmic dynamics and the interdimensional structure of multidimensional networks from displaying a snapshot-like behavior (as characterized by any arbitrary mathematical theory). Thus, beyond methods based on statistics or probability as traditionally seen in random graphs and complex networks models, representational incompressibility implies a necessary underlying constraint in the multidimensional network topology. We argue that the study of how isomorphic transformations and their respective algorithmic information distortions can characterize sequential interdimensional structures in (multidimensional) networks helps the analysis of network topological properties while being agnostic to the chosen theory, algorithm, computation model, and programming language.

\end{abstract}




\maketitle

\newtheorem{lemma}{Lemma}[section]

\newtheorem{theorem}[lemma]{Theorem}

\newtheorem{corollary}[lemma]{Corollary}

\theoremstyle{definition}
\newtheorem{definition}{Definition}[section]
\newtheorem{notation}[definition]{Notation}

\theoremstyle{remark}
\newtheorem{note}{Note}[section]

\pagebreak
\pagestyle{style1}

\section{Introduction}\label{sectionIntro}

Complex networks science has been showing fruitful applications to the study of biological, social, and physical systems \cite{Barabasi2009a,Lewis2009,Zarate2019nat}.
This way, as the interest and pervasiveness of complex networks modeling and network analysis increase in graph theory and network science, proper representations of multidimensional networks into new extensions of graph-theoretical abstractions has become an important topic of investigation \cite{Kivela2014,Lambiotte2019,Michail2018}.
In a general sense, a multidimensional network (also known as high-order network) is any network that has additional representational structures.
For example, this is the case of dynamic  (i.e., time-varying) networks \cite{Rossetti2018b,Michail2015,Pan2011,Costa2015a}, multilayer networks \cite{Kivela2014,Boccaletti2014a,Domenico2013}, and dynamic multilayer networks \cite{Wehmuth2016b,Wehmuth2018a,Wehmuth2018}. 
In this regard, the general scope of this paper is to 
study the lossless incompressibility of multidimensional networks
as well as network topological properties in such generalized graph representations.

Unlike traditional methods from random graphs theory, such as the probabilistic method, and from data compression algorithm analysis, which is dependent on the chosen programming language, we present a theoretical investigation of algorithmic complexity and algorithmic randomness.
First, this enables worst-case compressibility and network complexity analyses that do not depend on the choice of programming language \cite{Zenil2018a,Morzy2017a}.
Secondly, it enables us to achieve mathematical proofs of the existence of certain network topological properties in generalized graphs, which are not necessarily generated or defined by stochastic processes \cite{Buhrman1999,Li1997}, and thus necessarily occur due to phase transitions at the asymptotic limit as the network size increases.

With this purpose, algorithmic information theory \cite{Li1997,Calude2002,Downey2010} has been giving us computably universal tools for studying data compression of individual objects \cite{Barmpalias2019,Sow2003,Zenil2018,Delahaye2012,Li1997}.
On the one hand, first note that the source coding theorem \cite{Cover2005} in classical information theory ensures that, as $ \left| V(G) \right| = n \to \infty $, every recursively labeled representation of a random graph $ G $ on $n$ vertices and edge probability $ p = 1/2 $ in the classical Erdős–Rényi model $ \mathcal{G}(n,p) $---from an independent and identically distributed stochastic process---is expected to be losslessly incompressible with probability arbitrarily close to one \cite{Li1997,Leung-Yan-Cheong1978}.
In this sense, as shown in \cite{Dehmer2011,Mowshowitz2012,Zenil2018a,Lewis2009}, approaches to network complexity based on classical (or statistical) information theory have presented useful tools to find, estimate, or measure underlying graph-theoretic topological properties in random graphs or in complex networks. 

On the other hand, for some graphs $G$ (with $n$ vertices) displaying maximal degree-sequence entropy~\cite{Zenil2017a} or exhibiting a Borel-normal distribution of presence or absence of edges~\cite{Becher2002,Becher2015,Becher2013}, the edge set $ E( G ) $ is computable and, therefore, is algorithmically compressible on a logarithmic order \cite{Li1997,Downey2010,Calude2002}.
That is, even though their inner structures might seem to be statistically homogeneous or to be following a uniform probability distribution, these graphs $G$ may be compressed into $ \mathbf{O}\left( \log(n) \right) $ bits and thus are very far from being incompressible objects.
Moreover, it has been shown that statistical estimations are not invariant to language description in general \cite{Zenil2018a}.
This comes from the reliance on probability distributions that require making a choice of feature of interest relevant to the measure of interest. 
For example, the study of the distribution of node degrees disregards other representations  of the same object (network) and its possible underlying generating mechanism. 
While some of these statistical approaches may be useful when a feature of interest is selected they can only capture the complexity of the representation and not the object. This is in contrast to universal measures of randomness such as algorithmic complexity whose invariance theorem guarantees that different representation will have convergent values \cite{Li1997,Downey2010,Calude2002} (up to a constant that only depends on the choice of the universal programming languages). 
Of course, the complication is how to achieve the estimation of those universal measures which by virtue of being universal are also semi-computable and their application require, therefore, a much higher degree of methodological care compared to those measures that are simply computable such as those based on traditional statistical approaches, e.g., entropy, or graph-theoretic approaches, e.g., node degree.
In this article, we apply the results on computable labeling and algorithmic randomness introduced in \cite{Abrahao2018darxivandreport,Buhrman1999,Zenil2018a,Khoussainov2014}.
In particular, we extend these ones for string-based representations of classical graphs to string-based representations of multiaspect graphs (MAGs).
MAGs are formal representations of dyadic (or $2$-place) relations between two arbitrary $n$-ary tuples \cite{Wehmuth2017,Wehmuth2016b} and have shown fruitful representational properties to network modelling \cite{Wehmuth2015a,Wehmuth2016b,Abrahao2018publishedAMS,Abrahao2017published}, to perform analysis of multidimensional networks \cite{Wehmuth2018a,Costa2015a,Wehmuth2018b,Wehmuth2018}, such as dynamic networks and multilayer networks, and to investigate algorithmic information distortions due to isomorphisms~\cite{Abrahao2020cAIDistortionsCN,Abrahao2021AIDistortionsEntropy}.

First, we compare the algorithmic complexity and incompressibility of snapshot-dynamic networks and multiplex networks with more general forms of dynamic networks and multilayer networks, respectively.
In turn, dynamic networks and multilayer networks are considered as distinct types of multidimensional (or high-order) networks.
Secondly, we demonstrate some multidimensional topological properties of incompressible general multidimensional networks.
To tackle these problems in the present paper, we apply a theoretical approach by putting forward definitions, lemmas, theorems, and corollaries.

In Section~\ref{sectionSnapshot}, we investigate the algorithmic randomness of snapshot-dynamic networks and multiplex networks through the calculation of the worst-case lossless compression of the characteristic string of the network.
This way, one can compare these two kinds of networks with more general forms of dynamic networks and multilayer networks, respectively.
%
We demonstrate that the presence of multidimensional network topological properties that may not correspond to underlying structural constraints in real-world networks, e.g., of snapshot-dynamic networks. 
In particular, we show the presence of transtemporal or crosslayer edges (i.e., edges linking vertices at non-sequential time instants or layers).

\section{Snapshot-like multidimensional networks}\label{sectionSnapshot}

This section presents a theoretical investigation of the consequences of the results in Appendix~\ref{subsectionRandomMAGs} to some of the common representations of dynamic networks and multilayer networks.
As we will explain and formalize in Section~\ref{subsectionMultiplex}, we choose a differential approach to the dynamic and multilayer case, so that both become particular cases of general multidimensional networks while keeping their own distinct physical interpretation of what each `dimension' (or aspect) \cite{Wehmuth2017,Wehmuth2016b,Wehmuth2018b} represents.
In this way, we first present the investigation of the algorithmic complexity of snapshot-based representations of dynamic networks.
Then, in Section~\ref{subsectionMultiplex}, we introduce the same kind of investigation for multiplex networks.

\subsection{Snapshot-dynamic networks}\label{subsectionSnapshotDynamic}

In the context of real-world complex networks analysis, one may highlight some important representation models of dynamic networks, such as, time-varying graphs (TVGs) \cite{Wehmuth2015a,Costa2015a,Wehmuth2016b,Wehmuth2017}, temporal networks (TNs) \cite{Rossetti2018b,Pan2011}, temporal graphs (TGs) \cite{Michail2015}, and snapshot networks (SNs) \cite{Rossetti2018b,Wehmuth2015a}.
%
In this direction, we follow the same unifying and universal approach in \cite{Wehmuth2015a} with the purpose of showing that a particular class of dynamic networks (in the case, snapshot-dynamic networks) displays less irreducible information content than a more general representation of dynamic networks such as TVGs.
However, studying the advantages and disadvantages of each representation model in terms of network analysis is not in our current scope.
Thus, we focus on studying a general snapshot-based representation of dynamic networks and its algorithmic randomness in relation to time-varying graphs (TVGs) and their algorithmic randomness. 
Note that TVGs are second order multiaspect graphs (MAGs) \cite{Wehmuth2015a,Wehmuth2018,Costa2015a}. 

The main idea is to: first, briefly discuss equivalences of some of the main representations of snapshot-like dynamic networks; secondly, study the algorithmic randomness of snapshot-dynamic networks, which can be represented by a particular class of TVGs;
and, then, compare with the algorithmic randomness of general undirected dynamic networks, which are arbitrary simple TVGs, i.e., second order simple MAGs.

Except for the cases in which the pertinence of the vertices in each time instant (and not only its connectivity in each time instant) do matter in the network analysis---see node-alignment below---, one can easily show that a snapshot-based representation as in \cite{Rossetti2018b} is equivalent to the snapshot-based representation in \cite{Wehmuth2015a}. 
To this end, note that, in \cite{Rossetti2018b}, a snapshot network is defined as a sequence of graphs $ G_i = ( V_i , E_i ) $ in the form $ \left( G_0 , \dots , G_{t_{max}} \right) $. 
On the other hand, in \cite{Wehmuth2015a}, a snapshot network is a TVG  composed of only \emph{spatial} edges, i.e., edges that connect two vertices at the same time instant only. 
In other words, a snapshot-like dynamic network in \cite{Wehmuth2015a} is a special case of dynamic network that can be solely represented by, for example, the main diagonal blocks of the adjacency matrix of the isomorphic graph to the TVG  in Figure~\ref{Figure_SpatialTVG}. 

\begin{figure}[!htb] 
	\centering
	\includegraphics[width=.6\linewidth]{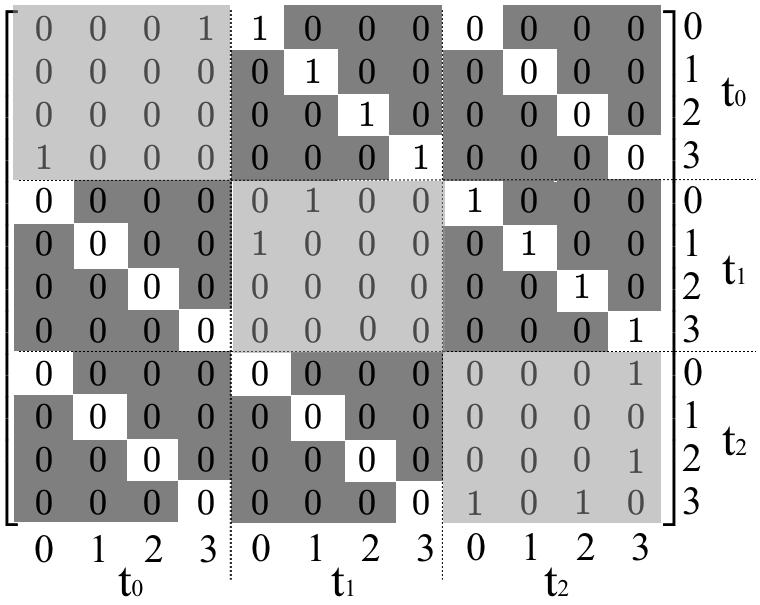}
	\caption{The adjancency matrix of the isomorphic graph to the TVG, which represents a sequentially coupled node-aligned dynamic network.} \label{Figure_SpatialTVG}
\end{figure}

For the sake of simplicity, we call an arbitrary TVG $ \mathrm{ G_t } = (\mathrm{V},\mathscr{E},\mathrm{T}) $ composed of only \emph{spatial} edges as a \emph{spatial} TVG. Therefore,
the sequence of vertex sets $ \left(  V_0, \dots , V_{t_{max}}  \right) $ in $ G_i = ( V_i , E_i ) $ may be mapped onto a larger vertex set 
\[ 
\mathrm{V} = \bigcup\limits_{ V_i \in \left(  V_0, \dots , V_{ t_{max} }  \right)  }^{} V_i
\]
\noindent such that $\mathrm{ G_t }=(\mathrm{V},\mathscr{E},\mathrm{T})$, $ \left| \{ G_i \vert \, G_i \in \left(  G_0, \dots , G_{ t_{max} } \right) \} \right| = \left| \mathrm{T}( \mathrm{ G_t } ) \right| $, and 
\begin{align}\label{BequationSpatialTVGandSnapshot}
	( u , v ) \in E_i( G_i ) \iff ( u , t_i , v , t_i ) \in \mathscr{E}( \mathrm{ G_t } ) 
	\text{ ,}
\end{align}
\noindent where $ 0 \leq i \leq t_{max} $.  
Inversely, a spatial TVG can be univocally represented by a sequence of graphs $ \left( G_0 , \dots , G_{t_{max}} \right) $ as in \cite{Rossetti2018b} by simply assuming $ V_i = V_j $ for every $ 0 \leq i \leq t_{max}  $ and $ 0 \leq j \leq t_{max}  $, so that the equivalence in Equation~\eqref{BequationSpatialTVGandSnapshot} also holds.

In most cases, a \emph{node-alignment}~\cite{Kivela2014,Cozzo2018} hypothesis (i.e., $ V_i = V_j $ for every $ 0 \leq i \leq t_{max}  $ and $ 0 \leq j \leq t_{max}  $) is assumed, so that no additional information would be needed to determine which vertices do not belong to a specific time instant.
Thus, for the present purposes of this article, we assume that the snapshot network is node-aligned.

Nevertheless, as studied in~\cite{Abrahao2020cAIDistortionsCN,Abrahao2021AIDistortionsEntropy} for arbitrary multidimensional networks instead of TVGs, an interesting future research would be to investigate the worst-case scenarios in which additional information is needed to recover the original snapshot network that is not node-aligned from the spatial TVG.
Indeed, an irreducible information dependency on the sequence $ \left(  \mathrm{V} \setminus V_0, \dots , \mathrm{V} \setminus V_{t_{max}}  \right) $ of excluded vertices may take place in a similar manner to the one on the companion tuple in \cite{Abrahao2018darxivandreport}.

In any event, note that retrieving the correspondent spatial TVG \cite{Wehmuth2015a,Wehmuth2016b} from the snapshot network \cite{Rossetti2018b} is always straightforward by the addition of \emph{empty nodes} \cite{Kivela2014}] (or, in MAG terminology, unconnected composite vertices \cite{Wehmuth2017}), since every TVG is defined on the set of composite vertices \cite{Wehmuth2016b} and, therefore, is always node-aligned by definition.

Another formalization of a snapshot-like dynamic network may be through restricting the set $ \mathbb{E}( \mathrm{ G_t } ) $ of all possible composite edges of a TVG $ \mathrm{ G_t } $ into another set $ \mathbb{E'}( \mathrm{ G_t } ) $ such that, for every $ e \in \mathbb{E}( \mathrm{ G_t } ) $ and $ i , j \in \mathbb{N} $,
\[
e = \left( u , t_i , v , t_ j \right) \in \mathbb{E'}( \mathrm{ G_t } )   \iff   j = f(i)
\text{ ,}
\]
\noindent where $ f  : \{  0 , f(0), f(f(0)), \dots , f^{-1}\left( \left| \mathrm{T}( \mathrm{ G_t } ) \right| - 1 \right) \}  \rightarrow  \{  f(0) , f(f(0)) , \dots , \left| \mathrm{T}( \mathrm{ G_t } ) \right| - 1 \}  $ is a strictly increasing bijective function defined on any recursive iteration. 
Thus, these restricted TVGs $ \mathrm{ G_t } = (\mathrm{V},\mathscr{E'},\mathrm{T}) $, where $ \mathscr{E'}\left( \mathrm{ G_t } \right) \subseteq \mathbb{E'}( \mathrm{ G_t } ) $, may present advantages when considering a more realistic scenario in which relations or communications between nodes are not instantaneous or demand non-equal time intervals over time. 
In addition, if one allows \emph{temporal edges} (i.e., edges connecting the same node at two distinct time instants ~\cite{Wehmuth2015a}), which are directly analogous to \emph{coupling edges} \cite{Kivela2014} in the multilayer case, some assumptions like the one that guarantees the transitivity on a node in the case it is disconnected within one or more snapshots become unnecessary.

In fact, unlike the multilayer case in which there may not be a physical interpretation of the necessity of preserving the pairwise ordering of layers---as we will also discuss in Section~\ref{subsectionMultiplex}---, this assumption of transitivity is represented by a restriction on the set of temporal (or coupling) edges:
we say a TVG is \emph{sequentially coupled} if all the temporal (or coupling) edges are connecting the same node $u$ from the time instant $t_i$ to the time instant $ t_{ i + 1 } $ only and, for every $ 0 \leq i \leq t_{max} $, every node $u$ has a temporal (or coupling) edge to itself from the time instant $t_i$ to the time instant $t_{ i + 1 }$.
The reader is invited to note that sequentially coupling is an even stronger restriction of the set of coupling edges than saying that the couplings are \emph{diagonal} and/or \emph{categorical} as in \cite{Kivela2014}---see also Section~\ref{subsectionMultiplex}.
That is, if the network is sequentially coupled, there is \emph{no} other temporal (or coupling) edge connecting a node $u$ at time instant $t_i$ to the same node $u$ at time instant $t_j$ than the case in which $j = i + 1 $.
An example of sequentially coupled networks are the temporal networks as defined in \cite{Boccaletti2014a}, should they be also node-aligned.

Thus, one can see that a snapshot-like dynamic network in the form $ (\mathrm{V},\mathscr{E'},\mathrm{T}) $ enables one to better represent the cases in which the sequential coupling does not hold in general; or, in other words, in which some nodes may not relay information for future communication in forthcoming time instants.
Related to this issue, it is a consequence of the result we will demonstrate in Appendix~\ref{sectionTopologicalTVG} that arbitrary incompressible simple TVGs are not sequentially coupled.
See Corollary~\ref{corTranstemporaledges}.

In any event, we have that this representation of snapshot-like dynamic networks in the form $ (\mathrm{V},\mathscr{E'},\mathrm{T}) $ can also be reduced to spatial TVGs $ (\mathrm{V},\mathscr{E},\mathrm{T''}) $---i.e., without the restrictions in the set of composite edges---by injectively mapping the set of time intervals onto another set $ \mathrm{T''}( \mathrm{ G_t } ) $ of time instants while preserving the previous ordering: 
for example, one runs a recursive bijective procedure that makes $ t''_i \equiv ( t_i , t_{ f(i) } ) $, where $ \left| \mathrm{T''}( \mathrm{ G_t } )  \right| \leq \left| \mathrm{T}( \mathrm{ G_t } )  \right| - 1 $, and
\begin{align*}\label{BequationSpatialTVGandIterativeSnapshot}
e = \left( u , t_i , v , t_{ f(i) } \right) \in \mathbb{E'}( \mathrm{ G_t } )  \iff ( u , t''_i , v , t''_i ) \in \mathscr{E}( \mathrm{ G_t } ) 
\text{ .}
\end{align*}
Again, as also may occur with snapshot networks $ \left( G_0 , \dots , G_{t_{max}} \right) $ that are not node-aligned, retrieving the spatial TVG from a snapshot-like dynamic network $ (\mathrm{V},\mathscr{E'},\mathrm{T}) $ is straightforward, whereas the inverse conversion may require additional information---in this latter particular case to determine which is the time interval $ ( t_i , t_{ f(i) } ) $ that each $ t''_i $ is representing (in the case these intervals are not uniformly equal).
Thus, a future investigation of the worst-case information dependency of a non-uniform function $f(i)$ will be necessary; and the irreducible topological information (i.e., the irreducible information necessary to determine/compute $ \left< \mathscr{E'} \right> $) carried by an arbitrary $ (\mathrm{V},\mathscr{E'},\mathrm{T}) $ may be less compressible than that of spatial TVGs.

In this way, for the purposes of this article, we assume hereafter the representation of a snapshot-like dynamic network as a spatial TVG; and we call them simply as \emph{snapshot-dynamic network}.
Note that underlying properties in snapshot-dynamic networks, like the sequential coupling, are fixed. 
Therefore, the algorithmic information of a spatial TVG and the algorithmic information of a spatial TVG with the addition of sequential couplings can only differ by a constant (that only depends on the chosen language $ \mathbf{L_U} $) and, thus, is negligible in our forthcoming results.

It is straightforward to calculate the maximum number of spatial directed edges $ e \in \mathscr{E}( \mathrm{ G_t } ) $ in a \emph{traditional} TVG $ \mathrm{ G_t^d } = ( \mathrm{V} , \mathscr{E} , \mathrm{T} ) $. 
Note that a \emph{traditional} MAG is a (directed or undirected) MAG without (composite) self-loops \cite{Abrahao2018darxivandreport,Wehmuth2016b}. 
(See Definition~\ref{defTraditionalMAG}).  
Moreover, from \cite{Wehmuth2015a,Costa2015a}, we have that a TVG is a second order MAG. (See Definition~\ref{BdefTVG}).
Therefore, we define a \emph{traditional directed} TVG $ \mathrm{ G_t^d } = ( \mathrm{V} , \mathscr{E} , \mathrm{T} ) $ as a TVG without 
self-loops. 
This way, we will have that there are
\[
\left| \mathrm{T}( \mathrm{ G_t^d } ) \right| \left(  \left| \mathrm{V}( \mathrm{ G_t^d } ) \right|^2 - \left| \mathrm{V}( \mathrm{ G_t^d } ) \right| \right)
\]
\noindent possible spatial directed edges. 

From a simple graph (i.e., an undirected graph without self-loops) perspective, we may consider spatial TVGs as sequences of simple graphs. We call a \emph{simple} TVG $ \mathrm{ G_t^u } = ( \mathrm{V} , \mathscr{E} , \mathrm{T} ) $ as a particular case of a simple second order MAG, where a simple MAG (see Definition~\ref{defSimplifiedMAG})) is defined in \cite{Abrahao2018darxivandreport} as a traditional \emph{undirected} MAG. This way, 
we will have that there are
\[
\left| \mathrm{T}( \mathrm{ G_t^u } ) \right| \left( \frac{  \left| \mathrm{V}( \mathrm{ G_t^u } ) \right|^2 - \left| \mathrm{V}( \mathrm{ G_t^u } ) \right| }{ 2 } \right)
\]
\noindent possible spatial undirected edges in \emph{simple spatial} TVGs. 

Now, let $ \mathrm{G'_t} $ denote an arbitrary simple spatial TVG $ \mathrm{ G_t^u } = ( \mathrm{V} , \mathscr{E} , \mathrm{T} ) $ that belongs to a \emph{recursively labeled} infinite family $ F_{ G'_t } $ of all possible simple TVGs with vertex labels in $ \mathbb{N} $. 
The existence of such a family is guaranteed by Lemma~\ref{lemmaLabeledfamilyofMAG}. 
Directly from Definition~\ref{BdefCharacteristicstringofasimpleMAG}, we have the binary string that univocally represents the presence or absence of an edge in $ \mathscr{E}( \mathrm{ G'_t } ) $; and we call it the \emph{characteristic string} of $ \mathrm{G'_t} $ \cite{Abrahao2018darxivandreport}.
In this sense, from Corollary~\ref{corFamilyoflabeledMAGandstrings}, one can see that a characteristic string promptly contains all the information that is necessary to computably retrieve the entire $ \mathrm{G'_t} $, except for the information required to apply a previously known recursive way to label the composite vertices and order the composite edges.

It is also important to note that the information encoded in the characteristic string may be displaying a decompressed form of its algorithmic information content.  
To tackle this issue, we define the algorithmic-informational version of the characteristic string, and thus formalizing such a notion of topological (algorithmic) information---see Theorem~\ref{thmSimplespatialTVG}, Corollary~\ref{corFamilyoflabeledMAGandstrings}, and Appendix~\ref{subsectionAIT}---that is potentially agnostic with respect to node labeling or indexing:

\begin{definition}\label{defACS}
	Let $  F_{ \mathscr{G}_c }   $ be a recursively labeled family of simple MAGs $ \mathscr{G}_c $.
	Let $ p''_1 , \, p''_2 \in \mathbf{L_U} $ be fixed and only depend on the choice of the family $  F_{ \mathscr{G}_c }   $.
	We say a binary string $ y \in \{ 0 , 1 \}^* $ is an \emph{algorithmically characteristic string} of $ \mathscr{G}_c $ with respect to $  F_{ \mathscr{G}_c }   $ \textit{iff} 
			$ \mathbf{U}( \left< y , p''_1 \right> ) = \left< \mathscr{E}( \mathscr{G}_c  ) \right> $ and 
			$ \mathbf{U}( \left< \left< \mathscr{E}( \mathscr{G}_c ) \right> , p''_2 \right> ) = y $.
\end{definition}

Thus, if $y$ is such an algorithmically characteristic string, it is immediate to show that 
\begin{equation*} 
	C( \mathscr{E}( \mathscr{G}_c  ) \vert y ) \leq K( \mathscr{E}( \mathscr{G}_c  )  \vert y ) + \mathbf{O}(1) = \mathbf{O}(1)
\end{equation*}
and 
\begin{equation*}
	C( y \vert \mathscr{E}( \mathscr{G}_c  ) ) \leq K( y \vert \mathscr{E}( \mathscr{G}_c  ) ) + \mathbf{O}(1) = \mathbf{O}(1) 
\end{equation*}
hold independently of the choice of $ \mathscr{G}_c $ in $ F_{ \mathscr{G}_c } $.

On the other hand, it may not be the case that the opposite implication $\big($i.e., 
\begin{equation*} 
	K( \mathscr{E}( \mathscr{G}_c  )  \vert y ) + \mathbf{O}(1) = \mathbf{O}(1)
\end{equation*}
and 
\begin{equation*}
	K( y \vert \mathscr{E}( \mathscr{G}_c  ) ) + \mathbf{O}(1) = \mathbf{O}(1) 
\end{equation*}
\noindent implying the existence of constants $ p''_1 , \, p''_2 \in \mathbf{L_U} $ such that 
$ \mathbf{U}( \left< y , p''_1 \right> ) = \left< \mathscr{E}( \mathscr{G}_c  ) \right> $ 
and 
$ \mathbf{U}( \left< \left< \mathscr{E}( \mathscr{G}_c ) \right> , p''_2 \right> ) = y $ $\big)$ does hold in general---and this should be an interesting future research. 
For example, a possible question in this direction might be whether it is possible or not to construct a recursively labeled infinite family of MAGs in which there is an infinite subfamily of MAGs that are K-trivial \cite{Downey2010}, but not computable, with respect to a string $y$.

In any event, from the proof of Corollary~\ref{corFamilyoflabeledMAGandstrings} presented in \cite{Abrahao2018darxivandreport}, we have that Definiton~\ref{defACS} is always satisfiable by taking the algorithmically characteristic string $y$ as, for instance, the very characteristic string.
This holds because of the recursively labeling of the entire family of MAGs, as in Definition~\ref{defLabeledfamilyofMAG}.
However, since one surely knows there are 
\[
\left( \frac{  \left( \left| \mathrm{V}( \mathrm{G'_t} ) \right| \, \left| \mathrm{T}( \mathrm{G'_t} ) \right| \right)^2 - \left| \mathrm{V}( \mathrm{G'_t} ) \right| \, \left| \mathrm{T}( \mathrm{G'_t} ) \right| }{ 2 } \right)
-
\left| \mathrm{T}( \mathrm{G'_t} ) \right| \left( \frac{  \left| \mathrm{V}( \mathrm{G'_t} ) \right|^2 - \left| \mathrm{V}( \mathrm{G'_t} ) \right| }{ 2 } \right)  
\] 
non-existent non-spatial undirected edges (including non-spatial undirected edges that are sequential couplings) and $ \mathrm{G'_t} $ is a second order simple  MAG $ \mathscr{G}_c $, one can compress the characteristic string of $ \mathrm{G'_t} $ in such a way that the resulting algorithmically characteristic string retains the algorithmic information carried, or conveyed, by the usual characteristic string:

\begin{theorem}\label{thmSimplespatialTVG}
	Let $ \mathrm{G'_t} = (\mathrm{V},\mathscr{E},\mathrm{T}) $ be a simple spatial TVG that belongs to a recursively labeled infinite family $ F_{ G'_t } $ of simple TVGs.
	Then, there is a binary string $ y  \in \{ 0 , 1 \}^*  $ that is an \emph{algorithmically characteristic string} of $ \mathrm{G'_t} $ such that
	\[
	K( y ) 
	\leq
	l(y) 
	+ \mathbf{O}(1)
	\leq 
	\left| \mathrm{T}( \mathrm{G'_t} ) \right| \left( \frac{  \left| \mathrm{V}( \mathrm{G'_t} ) \right|^2 - \left| \mathrm{V}( \mathrm{G'_t} ) \right| }{ 2 } \right)  
	+ 
	\mathbf{O}\left( \log_2\left( \left| \mathrm{V}( \mathrm{G'_t} ) \right| \, \left| \mathrm{T}( \mathrm{G'_t} ) \right| \right) \right)
	\text{ ,}
	\]
	\[
	K( x ) 
	\leq 
	K(y) + \mathbf{O}( 1 )
	\text{ ,}
	\]
	\[
	K( x  \mid {y}^* ) \leq K( x  \mid {y} ) + \mathbf{O}\left( 1 \right) \leq \mathbf{O}\left( 1 \right)
	\text{ ,}
	\]
	\[
	K( y ) \leq K(x) + \mathbf{O}( 1 )
	\text{ ,}
	\]
	and
	\[
	K( y \mid x^* ) \leq K( y \mid x ) + \mathbf{O}( 1 ) \leq \mathbf{O}( 1 )
	\]
	hold, where $x$ is the characteristic string of $ \mathrm{G'_t} $.

\end{theorem}

Thus, for every simple spatial TVG  $ \mathrm{G'_t} = (\mathrm{V},\mathscr{E},\mathrm{T}) $ that belongs to a \emph{recursively labeled} infinite family $ F_{ G'_t } $ of simple TVGs, we will have that, from Lemma~\ref{lemmaBasicAIT}, Theorem~\ref{thmSimplespatialTVG}, and Corollary~\ref{corFamilyoflabeledMAGandstrings},
\begin{equation*}
	\begin{aligned}
		C\left( \mathscr{E}( \mathrm{G'_t} ) \, \Big\vert \, \left( \left| \mathrm{V}( \mathrm{G'_t} ) )  \right| \, \left| \mathrm{T}( \mathrm{G'_t} ) )  \right| \right) \right) 
		& \leq 
		K( \mathscr{E}( \mathrm{G'_t} ) ) +  \mathbf{O}( 1 ) 
		\leq \\
		& \leq 
		\left| \mathrm{T}( \mathrm{G'_t} ) \right| \left( \frac{  \left| \mathrm{V}( \mathrm{G'_t} ) \right|^2 - \left| \mathrm{V}( \mathrm{G'_t} ) \right| }{ 2 } \right)  \\  
		& \quad + \mathbf{O}\left( \log_2\left( \left| \mathrm{V}( \mathrm{G'_t} ) \right| \, \left| \mathrm{T}( \mathrm{G'_t} ) \right| \right) \right)
	\end{aligned}
\end{equation*}
\noindent holds.\footnote{ In fact, one can even improve this inequality to show that $ C\left( \mathscr{E}( \mathrm{G'_t} ) \, \Big\vert \, \left( \left| \mathrm{V}( \mathrm{G'_t} ) )  \right| \, \left| \mathrm{T}( \mathrm{G'_t} ) )  \right| \right) \right) \leq \mathbf{O}\left( \left| \mathrm{T}( \mathrm{G'_t} ) \right| \left( \frac{  \left| \mathrm{V}( \mathrm{G'_t} ) \right|^2 - \left| \mathrm{V}( \mathrm{G'_t} ) \right| }{ 2 } \right) \right) $.}
On the other hand, as shown in \cite{Abrahao2018darxivandreport}, we directly have from Lemma~\ref{lemmaK-randomMAGs}, Corollary~\ref{corFamilyoflabeledMAGandstrings}, and Theorem~\ref{thmK-randomandC-randomMAGs} that, for arbitrary $ \mathbf{O}(1) $-K-random simple TVGs $ \mathrm{ G_t } $,
\[
K( \mathscr{E}( \mathrm{G_t} ) ) 
\geq
\left( \frac{  \left( \left| \mathrm{V}( \mathrm{G_t} ) \right| \, \left| \mathrm{T}( \mathrm{G_t} ) \right| \right)^2 - \left| \mathrm{V}( \mathrm{G_t} ) \right| \, \left| \mathrm{T}( \mathrm{G_t} ) \right| }{ 2 } \right)
-  
\mathbf{O}( 1 ) 
\]
\noindent and 
\begin{equation*}
\begin{aligned}
	C( \mathscr{E}( \mathrm{G_t} ) \mid \left( \left| \mathrm{V}( \mathrm{G_t} ) )  \right| \, \left| \mathrm{T}( \mathrm{G_t} ) )  \right| \right) ) 
	\geq
	\left( \frac{  \left( \left| \mathrm{V}( \mathrm{G_t} ) \right| \, \left| \mathrm{T}( \mathrm{G_t} ) \right| \right)^2 - \left| \mathrm{V}( \mathrm{G_t} ) \right| \, \left| \mathrm{T}( \mathrm{G_t} ) \right| }{ 2 } \right) \\
	-
	\mathbf{O}\left( \log_2\left( \left| \mathrm{V}( \mathrm{G_t} ) \right| \, \left| \mathrm{T}( \mathrm{G_t} ) \right| \right) \right)
	\text{ .}
\end{aligned}
\end{equation*}
Therefore, for large enough sets of time instants, $ \mathbf{O}(1) $-K-random simple TVGs carry at least
\begin{equation}\label{eqSnapshotnetwork}
\begin{aligned}
	\left( \frac{  \left( \left| \mathrm{V}( \mathrm{G_t} ) \right| \, \left| \mathrm{T}( \mathrm{G_t} ) \right| \right)^2 - \left| \mathrm{V}( \mathrm{G_t} ) \right| \, \left| \mathrm{T}( \mathrm{G_t} ) \right| }{ 2 } \right) 
	-
	\left| \mathrm{T}( \mathrm{G_t} ) \right| \left( \frac{  \left| \mathrm{V}( \mathrm{G_t} ) \right|^2 - \left| \mathrm{V}( \mathrm{G_t} ) \right| }{ 2 } \right) \\
	-
	\mathbf{O}\left( \log_2\left( \left| \mathrm{V}( \mathrm{G_t} ) \right| \, \left| \mathrm{T}( \mathrm{G_t} ) \right| \right) \right)
\end{aligned}
\end{equation}
more topological (or network) irreducible information than any simple spatial TVG could carry.

\subsection{Multiplex networks}\label{subsectionMultiplex}

Other network models of increasing importance in network science are those in which the nodes and/or connections between nodes may have distinct features \cite{Kivela2014,Boccaletti2014a,Lambiotte2019,Domenico2015,Contractor2011,Wehmuth2016b} other than the time dependency.
For example, one may differentiate in the case of: 
\begin{itemize}
	\item social networks, 
	\begin{itemize}
		\item between connection features, such as friendship, family, professional colleagues, face-to-face interaction, email interaction, etc \cite{Kivela2014,Contractor2011,Cozzo2018}; 
		\item or between node features, such as gender, online platform, school, city, company, etc \cite{Contractor2011,Kivela2014};
	\end{itemize} 
	\item biological networks, 
	\begin{itemize}
		\item between connection features, such as distinct nature of protein-protein reactions \cite{Domenico2015}, neuronal interactions (in particular, either through chemical or ionic channels) \cite{Boccaletti2014a}, etc; 
		\item or between node features, such as distinct species, gene pool, community, location \cite{Pilosof2017}, etc;
	\end{itemize} 
	\item multimodal transportation networks, 
	\begin{itemize}
		\item between connection features, such as bus network, the subway network, the air transportation network~\cite{Lambiotte2019,Wehmuth2018a}, etc; 
		\item or between node features, such as airline companies \cite{Wehmuth2018a}, etc.
	\end{itemize}
\end{itemize}
Usually in the modeling of complex networked systems, each of these features has been called a \emph{layer} \cite{Kivela2014,Boccaletti2014a,Wehmuth2017,Wehmuth2018a}.
Particularly, in this section, we focus on connection features, analyzing the theoretical characteristics of the so called \emph{multiplex} networks~\cite{Cozzo2018,Kivela2014}.
A multiplex network is a type of multilayer network.
We first briefly discuss some representation equivalences.
Then, we apply an analogous investigation to that of Section~\ref{subsectionSnapshotDynamic}. 

Before presenting multiplex networks, it is important to recover some definitions and nomenclatures from the previous literature on the subject, clarifying conditions or assumptions behind the mathematical concepts, so as to enable unambiguous applications of our theoretical results in future studies of real-world networks.
In doing so, we are not only grounding our nomenclature on, but also following the same purpose of terminology unification in \cite{Kivela2014}.
In this way, we choose an approach in order that one can distinguish between multilayer networks and dynamic networks, even though both can be viewed as just distinct physical interpretations of the same kind of network's aspect (or dimension).

To this end, we compare a well-known definition of multilayer network as in \cite{Kivela2014} with the multidimensional (i.e., high-order) approach formalized in \cite{Wehmuth2016b,Wehmuth2017}.
In \cite{Kivela2014} (or, equivalently, in \cite{Boccaletti2014a}), a \emph{multilayer} network $ M =
(V_M , E_M , V, \mathbf{L} ) $ is understood as an interconnected and/or intraconnected labeled collection of graphs.
Each of these graphs represents a `layer' whose (distinct or not) vertex sets can be either linked within this `layer' or linked to a vertex in another `layer'.  
Formally, a multilayer network $ M =
(V_M , E_M , V, \mathbf{L} ) $ \cite{Kivela2014} is defined by: 
\begin{itemize}
\item $ V $ denotes the set of all possible vertices $ v $; 
\item $ \mathbf{L} = \{ L_a \}_{a=1}^d $ denotes a collection of $ d \in \mathbb{N} $ sets $ L_a $ composed of \emph{elementary layers} $ \alpha \in  L_a $;
\item $ V_M \subseteq V \times L_1 \times \cdots \times L_d $ denotes the subset of all possible vertices that belong to a `layer' $ L_1 \times \cdots \times L_d $;
\item $ E_M \subseteq V_M \times V_M $ denotes the set of \emph{interlayer} and/or \emph{intralayer} edges connecting two \emph{node-layer tuples} $ \left( v , \alpha_1 , \dots , \alpha_d \right) \in V_M $.
\end{itemize}
An interlayer edge is defined as $ \left( \left( u , \alpha_1 , \dots , \alpha_d \right) , \left( v , \beta_1 , \dots , \beta_d \right) \right) \in E_M $ such that $ \left( \alpha_1 , \dots , \alpha_d \right) \neq \left( \beta_1 , \dots , \beta_d \right) $ and an intralayer is defined as $ \left( \left( u , \alpha_1 , \dots , \alpha_d \right) , \left( v , \beta_1 , \dots , \beta_d \right) \right) \in E_M  $ such that $ \left( \alpha_1 , \dots , \alpha_d \right) = \left( \beta_1 , \dots , \beta_d \right) $.
In addition, one defines a \emph{coupling} edge $ \left( \left( u , \alpha_1 , \dots , \alpha_d \right) , \left( v , \beta_1 , \dots , \beta_d \right) \right) \in E_M $ as an interlayer edge with $ u = v $.

Actually, there are in fact some equivalences between a MAG \cite{Wehmuth2016b} (see Definition~\ref{defMAG}) and a multilayer $ M =
(V_M , E_M , V, \mathbf{L} ) $:
$ V $ is the usual set of vertices $ V( \mathscr{G} )  \equiv  \mathscr{A}( \mathscr{G} )[1] $;
each set $ L_a $ is the $a$-th aspect $ \mathscr{A}( \mathscr{G} )[a]  $ of a MAG $ \mathscr{G} $;
$ V_M $ is a subset of the set of all composite vertices $ \mathbb{V}( \mathscr{G} ) $, where a node-layer tuple $ \left( v , \alpha_1 , \dots , \alpha_d \right) \in V_M $ is a composite vertex $ \mathbf{v} \in \mathbb{V}( \mathscr{G} ) $;
$ E_M \equiv  \mathscr{E}( \mathscr{G} ) $ is a subset of the set of all composite edges $ \mathbb{E}(\mathscr{G})  $.
The only distinctive characteristic of $ M =
(V_M , E_M , V, \mathbf{L} ) $ and $ \mathscr{G}=(\mathscr{A},\mathscr{E}) $ is the possibility that one or more vertices do not belong to one or more aspects of a MAG. 
Therefore, if a multilayer network $M$ is node-aligned \cite{Kivela2014,Cozzo2018}, i.e., $ V_M = V \times L_1 \times \cdots \times L_d $, this multilayer network $M$ is equivalent to a $ \left( d+1 \right) $-order MAG.
Thus, a node-aligned multilayer network with $d$ layers is a network that can be mathematically represented by a $ \left( d+1 \right) $-order MAG.

In most cases, as mentioned in Section~\ref{subsectionSnapshotDynamic}, a node-alignment hypothesis can be assumed \cite{Kivela2014}, except for the cases in which the pertinence of the vertices in each layer (and not only its connectivity) do matter in the network analysis.
In this regard, we also leave as future research the investigation of the worst cases for the information needed to retrieve $ V_M $ from $ V \times L_1 \times \cdots \times L_d $.
Due to this possibility, the algorithmic information carried by $M$ may be larger than that of the correspondent MAG $ \mathscr{G} $ with $d+1$ aspects and $ E_M =  \mathscr{E}( \mathscr{G} ) $.
Anyway, for the purposes of the present article (as we did in the dynamic case in Section~\ref{subsectionSnapshotDynamic}), we assume hereafter that the multilayer networks $M$ are node-aligned, so that $ \left< E_M \right> $ denotes the string $ \left< \mathscr{E}( \mathscr{G} )  \right>  $ with respect to such correspondent MAG $ \mathscr{G} $, where $  E_M \coloneqq  \mathscr{E}( \mathscr{G} )  $.

Besides the above formalities, we choose a distinguishing interpretation of the two currently studied multidimensional structures in complex networks: dynamism and ``multilayerism''.\footnote{ And this distinction will be reflected in our nomenclature and in our notation.}
In accordance with \cite{Wehmuth2016b,Wehmuth2017,Wehmuth2018a}, and unlike \cite{Kivela2014,Boccaletti2014a,Domenico2016} (where dynamic networks are considered as a particular type of multilayer networks), we are considering a time-varying graph topology, like in dynamic networks in Section~\ref{subsectionSnapshotDynamic}, and a multilayer topology, like the ones of the above described $M$, as two distinct multidimensional structures.

At a first glance, both the mathematical representations of a finite time progression and of a finite number of different layers can be performed by an ordered set of labels (or indexes).
However, besides distinct physical properties, a multilayer network may not need to obey a sequential indexing of layers that corresponds to a meaningful ordering of the physical counterparts of each layer.
This promptly differs from the sequential coupling introduced in Section~\ref{subsectionSnapshotDynamic}.
For example, node features (e.g., species, gender, or company) or connection features (e.g., bus network or email interaction) in the multilayer case do not have an intrinsic underlying structure that indicates the direction that the information is `flowing through' an edge, whereas the opposite holds in principle for time instants in dynamic networks.
Moreover, as we will show in Appendix~\ref{sectionTopologicalTVG}, some topological properties, such as the presence of crosslayer edges, may happen to make sense in multilayer networks, whereas, in dynamic networks---specially, in snapshot-dynamic networks---the presence of transtemporal edges may not.

In addition, we are assuming a general meaning of the nomenclature `multidimensional' so as to encompass both each individual node feature\footnote{ See the first paragraph of this section for more examples of these individual features.} and each type of the individual node features.
We say that each type of individual feature, such as being an arbitrary set of time instants or being an arbitrary set of layers, is a \emph{node dimension}.
Thus, a node dimension corresponds to an aspect of a MAG, as in \cite{Wehmuth2016b,Wehmuth2017,Wehmuth2018a}.
This is in consonance with a common understanding of a dimension as being an aspect or property in which a particular object can assume different values, names, etc. 

However, note that this differs from some usages of the term in the literature of network science, where for example one may say that a node $u$ linked to a node $v$ through a family relation lies on a different `dimension' than that of a node $u'$ linked to a node $v'$ through a professional relation.
In this sense, the particular object that can assume different values is the connection itself and the different values are the nodes. 
This way, any element of an aspect of a MAG other than the set of vertices represents a `dimension'.
In fact, such usage of the term may be also found in an overlap with that of multiplex networks \cite{Boccaletti2014a,Berlingerio2011}.
Thus, in order to avoid ambiguities, we call this kind of dimension as \emph{connection dimension}.
As a consequence, it derives directly from our chosen nomenclature that every connection dimension belongs to a node dimension.

In this way, our approach ensures that in any case one employs multidimensional networks, it can either mean networks with more than one node dimension, e.g., a dynamic network, or networks with more than one connection dimension, e.g., a social network with only two types of interactions, like family and professional.
Thus, unlike for example in \cite{Boccaletti2014a,Berlingerio2011}, in which multidimensional networks refers basically to multiplex networks, we adopt the convention of defining a \emph{general multidimensional network} as a network that has more than one node dimension and more than one connection dimension.
As a consequence, both multilayer networks and dynamic networks become particular cases of general multidimensional networks.
Also note that a network with only one node dimension and connection dimension is a monoplex (i.e., single-layer or monolayer) network \cite{Kivela2014,Domenico2013,Cozzo2018} and, consequentially, is totally equivalent to a graph.

In addition, we define a \emph{high-order network} as a network that can be mathematically represented by a \emph{high-order} MAG, i.e., a MAG with two or more aspects, each containing two or more elements \cite{Wehmuth2018b,Wehmuth2018a,Abrahao2018darxivandreport}.
Therefore, every node-aligned general multidimensional network is a high-order network.

In fact, it is shown in \cite[Section 3.3]{Wehmuth2017} that one can, for instance, construct a main-component graph $ m( \mathscr{G} ) $, which is defined as the MAG $ \mathscr{G} ) $ with the unconnected composite vertices excluded.
In this sense, if one allows the a priori exclusion of arbitrary composite vertices, it is possible to define MAGs that are not node-aligned, which would establish a complete equivalence between general multidimensional networks and high-order networks according to our nomenclature.
However, in consonance with Appendix~\ref{appendixBackgroundMAGs} and \cite{Abrahao2018darxivandreport}, we choose to stick with the notion of a MAG defined on a full set of possible composite vertices in this article. 
Hereafter, unless specified differently, multidimensional networks refers to node-aligned general multidimensional (i.e., high-order) networks and dimension refers to a node dimension.
Whereas we are focusing on the multiplex case in this section, we will return to higher-order network properties in Appendix~\ref{sectionTopologicalTVG}.

As in \cite{Kivela2014,Cozzo2018}, and similarly in \cite{Domenico2013}, a \emph{multiplex network} $ \mathcal{M} $ is a particular case of multilayer networks $ M =
(V_M , E_M , V, \mathbf{L} ) $ that are diagonally coupled, categorical, and potentially layer-connected, where $ \mathbf{L} =  \{ L_a \}_{a=1}^1 = \{L_1\} $ and $ \left| L_1 \right| \geq 2 $.
From \cite{Kivela2014}, we have that: 
a network $ M $ is \emph{diagonally coupled} if and only if, for every interlayer edge 
	$ \left( \left( u , \alpha \right) , \left( v , \beta \right) \right) \in E_M $, where $  \alpha  \neq  \beta  $, one has that $ u = v $;
a network $M$ is \emph{categorically coupled} if and only if, for every $ \left( u , \alpha \right) , \left( u , \beta \right) \in V_M $, one has that $ \left( \left( u , \alpha \right) , \left( u , \beta \right) \right) \in E_M $.
Also in consonance with \cite{Kivela2014}, we define here a condition for multiplex networks in order to ensure that the categorical couplings always apply to each pair of layers for at least one vertex $ u \in  V $: a network $M$ is \emph{potentially layer-connected} if and only if one has that
$ V_{  \alpha } \cap V_{ \beta } \neq \emptyset $, where $ V_{  \gamma } \coloneqq  \left\{ v \, \middle\vert\,  \left( v , \gamma \right) \in V_M \right\}$ and $ \alpha, \beta, \gamma \in L_1 $ are arbitrary.
This property is particularly important if $M$ is not node-aligned.
In general, most multiplex networks are considered to be node-aligned \cite{Kivela2014}.
In fact, for example in \cite{Boccaletti2014a,Radicchi2017}, multiplex networks are defined already assuming a node-alignment hypothesis.
Moreover, as we already saw for node-aligned multilayer networks, we will have that any node-aligned multiplex network $ \mathcal{M} $ is equivalent to a particular type of second order MAG.
Hereafter, unless specified differently, we will consider only node-aligned multiplex networks.

Let $ G_{ \mathcal{M} } $ be the second order simple MAG that is equivalent to an undirected node-aligned multiplex network $ \mathcal{M} $.
As in Section~\ref{subsectionSnapshotDynamic}, 
note that diagonal coupling and categorical coupling are fixed.\footnote{ See the construction of program $ s_1 $ in the proof of Theorem~\ref{thmSimplespatialTVG}.}
Let $  \mathrm{G'_{ \mathcal{M} }}  = ( \mathrm{V}, \mathscr{E}, L_1 ) $ denote the second order simple MAG $ G_{ \mathcal{M} } $ with all the interlayer edges excluded.
In addition, $  \mathrm{G'_{ \mathcal{M} }} $ belongs to a \emph{recursively labeled} infinite family $ F_{ G'_{ \mathcal{M} } } $ of arbitrary simple second order MAGs.
Then, analogously to Section~\ref{subsectionSnapshotDynamic}, the algorithmic information of $ G_{ \mathcal{M} } $ and the algorithmic information of $ \mathrm{G'_{ \mathcal{M} }} $ can only differ by a constant (that only depends on the chosen language $ \mathbf{L_U} $) and, thus, will be negligible in our forthcoming results.

With these definitions and nomenclature clarified in this section---as the reader may notice---, $ \mathrm{G'_{ \mathcal{M} }} $ is in fact an equivalent representation of a spatial simple TVG, except for a change in notation and nomenclature.
Therefore, we can now directly translate Theorem~\ref{thmSimplespatialTVG} and all the other results from Section~\ref{subsectionSnapshotDynamic}:

\begin{corollary}\label{corMultiplex}
	Let $  \mathrm{G'_{ \mathcal{M} }}  = ( \mathrm{V}, \mathscr{E}, L_1 ) $ belong to a recursively labeled infinite family $ F_{ G'_{ \mathcal{M} } } $ of simple second order MAGs.
	Then, there is a binary string $ y  \in \{ 0 , 1 \}^*  $ that is an \emph{algorithmically characteristic string} of $ \mathrm{G'_{ \mathcal{M} }} $ such that
	\begin{equation*}
	\begin{aligned}
	K( y )
	\leq \\
	l(y) 
	+ \mathbf{O}(1)
	\leq \\
	\left| L_1( \mathrm{G'_{ \mathcal{M} }} ) \right| \left( \frac{  \left| \mathrm{V}( \mathrm{G'_{ \mathcal{M} }} ) \right|^2 - \left| \mathrm{V}( \mathrm{G'_{ \mathcal{M} }} ) \right| }{ 2 } \right)  
	+ 
	\mathbf{O}\left( \log_2\left( \left| \mathrm{V}( \mathrm{G'_{ \mathcal{M} }} ) \right| \, \left| L_1( \mathrm{G'_{ \mathcal{M} }} ) \right| \right) \right)
	\text{ ,}
	\end{aligned}
	\end{equation*}
	where
	\[
	K( x ) 
	\leq 
	K(y) + \mathbf{O}( 1 )
	\text{ ,}
	\]
	\[
	K( x  \mid {y}^* ) \leq K( x  \mid {y} ) + \mathbf{O}\left( 1 \right) \leq \mathbf{O}\left( 1 \right)
	\text{ ,}
	\]
	\[
	K( y ) \leq K(x) + \mathbf{O}( 1 )
	\text{ ,}
	\]
	and
	\[
	K( y \mid x^* ) \leq K( y \mid x ) + \mathbf{O}( 1 ) \leq \mathbf{O}( 1 )
	\]
	hold and $x$ is the characteristic string of $ \mathrm{G'_{ \mathcal{M} }} $.

\end{corollary}


Therefore, for large enough sets $ L_1 $, $ \mathbf{O}(1) $-K-random undirected node-aligned multilayer networks with $ \mathbf{L} = \{ L_1 \} $ (i.e., simple second order MAGs) carry at least
\begin{equation}\label{eqMultiplex}
\begin{aligned}
\left( \frac{  \left( \left| \mathrm{V}( \mathrm{G'_{ \mathcal{M} }} ) \right| \, \left| L_1( \mathrm{G'_{ \mathcal{M} }} ) \right| \right)^2 - \left| \mathrm{V}( \mathrm{G'_{ \mathcal{M} }} ) \right| \, \left| L_1( \mathrm{G'_{ \mathcal{M} }} ) \right| }{ 2 } \right) \\
- \left| L_1( \mathrm{G'_{ \mathcal{M} }} ) \right| \left( \frac{  \left| \mathrm{V}( \mathrm{G'_{ \mathcal{M} }} ) \right|^2 - \left| \mathrm{V}( \mathrm{G'_{ \mathcal{M} }} ) \right| }{ 2 } \right) \\
-
\mathbf{O}\left( \log_2\left( \left| \mathrm{V}( \mathrm{G'_{ \mathcal{M} }} ) \right| \, \left| L_1( \mathrm{G'_{ \mathcal{M} }} ) \right| \right) \right)
\end{aligned}
\end{equation}
more topological (or network) irreducible information than any node-aligned undirected multiplex network could carry.

Moreover, the reader is invited to note that (if both the set of vertices and the set of layers are large enough) Theorem~\ref{thmTransaspectedges} will imply that incompressible node-aligned general multilayer networks cannot be \emph{diagonally coupled} and, therefore, cannot be a multiplex networks.
To this end, just replace an aspect corresponding to a set of layers with the first aspect (which is the set of vertices) and vice-versa.

\subsection{Algorithmic-informational characterization of sequential network structures}\label{sectionAlgorithmicsnapshot}

We saw in Section~\ref{subsectionSnapshotDynamic} that snapshot-dynamic networks inevitably can carry only a number of bits of computably irreducible information that is upper bounded by $ \mathbf{O}\left( \left| \mathrm{T} \right| \left( \frac{  \left| \mathrm{V} \right|^2 - \left| \mathrm{V} \right| }{ 2 } \right)  \right) $. 
In Section~\ref{subsectionMultiplex}, we showed that the same also occurs for multiplex networks.
Thus, one can define a non-empty class of (undirected) networks whose topological information---not to be conflated with the topological structure itself---characterizes a snapshot-like structure, in particular for snapshot-dynamic networks or multiplex networks but that can be generalized for any multidimensional network:

\begin{definition}\label{defSnapshotlikeMAG}
	Let $ \mathscr{G'}_c $ be a simple second order MAG that belongs to a recursively labeled infinite family $ F_{ \mathscr{G'}_c } $ of simple second order MAGs.
	We say that the network mathematically represented by the MAG $ \mathscr{G'}_c $ is an (undirected) \emph{algorithmically snapshot-like multidimensional network} with respect to dimension (i.e., aspect) $ \mathscr{A}( \mathscr{G} )[i] $ if and only if there is an algorithmically characteristic string $ y = \left< y' , w \right> \in \{ 0 , 1 \}^* $ of $ \mathscr{G'}_c $ (as in Definition~\ref{defACS}) such that $ w \in \{ \emptyset \} \cup \{ 0 , 1 \}^* $ is independent of the choice of $ \mathscr{G'}_c \in F_{ \mathscr{G'}_c } $, one has it that
	$ y' = \left< x_1 , \dots , x_{ \left|  \mathscr{A}( \mathscr{G} )[i]  \right| } \right>  \in \{ 0 , 1 \}^* $
	and $ x_{\alpha} = \left< z_1 , \dots , z_{ k_j } \right> \in \{ 0 , 1 \}^* $, where $ 1 \leq \alpha \leq \left|  \mathscr{A}( \mathscr{G} )[i]  \right| $, $ k_j = \binom{ \left|  \mathscr{A}( \mathscr{G} )[j]  \right| }{ 2 } $, $ j \neq i $, and $ z_h \in \{ 0 , 1 \} $ with $ 1 \leq h \leq k_j $,
\end{definition}

In other words, an \emph{algorithmically snapshot-like} (undirected) \emph{multidimensional network} is a network that can be totally represented by a second order simple MAG whose composite edge set can be \emph{algorithmically} determined by only informing a sequence of presences or absences of composite edges connecting two nodes within the same node dimension.
Note that the proof of Theorem~\ref{thmSimplespatialTVG} and Corollary~\ref{corMultiplex} guarantees that Definition~\ref{defSnapshotlikeMAG} is \emph{satisfiable}, for example, by snapshot-dynamic networks or multiplex networks that are node-aligned.
The reader is also invited to note that Definition~\ref{defSnapshotlikeMAG} can be easily generalized for simple MAGs \emph{with more than two aspects} (i.e., for $ p > 2 $).
(We chose to present the second order case for sake of simplifying the introduction of new notation in the present article).

Now, let a \emph{transtemporal} edge be a composite edge $ e=( u , t_i , v , t_j ) \in \mathscr{E}( \mathrm{ G_t } ) $ with $ j \neq i \pm 1 $ and  $ j \neq i $. 
By generalizing the temporal case, let a \emph{crosslayer} edge be a composite edge $ e=( u ,  \dots  , x_{ki} , \dots ,  x_{ps} , v , \dots , x_{kj} \dots , x_{ps'}) \in \mathscr{E}( \mathscr{ G }_c  ) $ with $ j \neq i \pm 1 $ and  $ j \neq i $. 
In fact, in accordance with the previous Sections~\ref{subsectionSnapshotDynamic}, and~\ref{subsectionMultiplex}, both transtemporal and crosslayer edges are particular cases of what we call by a \emph{non-sequential interdimensional edge}, should the aspect $ \mathscr{A}( \mathscr{G}_c )[k] $ corresponds to an arbitrary node dimension of the \emph{general multidimensional network}.

It is straightforward to check from Definitions~\ref{defACS} and~\ref{defSnapshotlikeMAG} that if a multidimensional network is sufficiently compressible---e.g., its topology contains a sufficient amount of redundancies, symmetries, etc---, then there are examples of multidimensional networks, each of which containing \emph{non-sequential interdimensional edges}, although these networks are \emph{algorithmically} snapshot-like multidimensional networks.
This may indicate or reveal some \emph{dissonance} between our intuition about information and structural representation.
That is, informationally characterizing the sequential interdimensional structure of a network is not always the same as (or does not always corresponds to) the characterization in terms of the mathematical structure representation of it.

Indeed, this fundamental mathematical phenomena has been investigated also in \cite{Khoussainov2014,Harrison-Trainor2019} for infinite graphs and in~\cite{Abrahao2020cAIDistortionsCN,Abrahao2021AIDistortionsEntropy,Abrahao2023bSemanticrobustnessCLPMST} for networks in general. 
We argue that the study of such a dissonance between information and structural representation in mathematics in general, but specially focusing on networks, is paramount for understanding the underpinnings of formal mathematical theories capable of defining sequential dimensions, like time in physics, while being agnostic to human biases toward certain types of representations.

The following Theorem~\ref{thmSequentialstructure} demonstrates that once the algorithmic information content of a multidimensional network gets sufficiently high (i.e., the network topology is sufficiently incompressible), the notion of sequentially structured dimension becomes conflated with that of sequentially structured information.

\begin{theorem}\label{thmSequentialstructure}
	Let $ \mathscr{ G }_c $ be any large enough $ \delta\left( \left| \mathbb{V}( \mathscr{ G }_c ) \right| \right) $-K-random simple MAG with order $ p \geq 2 $ and
	\[ \delta\left( \left| \mathbb{V}( \mathscr{ G }_c ) \right| \right) = \mathbf{o}\left( 
	\frac{ \left| \mathbb{V}( \mathscr{ G }_c ) \right|^2 - \left| \mathbb{V}( \mathscr{ G }_c ) \right| }{ 2 }
	- \left| \mathscr{A}( \mathscr{ G }_c  )[k]  \right| \left( \frac{ \left| \mathbb{ V }( \mathscr{ G }_c  ) \right| }{ \left| \mathscr{A}( \mathscr{ G }_c  )[k]  \right| } \right)^2
	\right) \text{ ,} \]
	where
	\[
	\left| \mathscr{A}( \mathscr{ G }_c  )[k]  \right|
	=
	\frac{\left| \mathbb{ V }( \mathscr{ G }_c  ) \right| }{ \left| \mathrm{ V }( \mathscr{ G }_c ) \right| \, 
		\bigtimes\limits_{ 2 \leq h \leq p , h \neq k \leq p} \left| \mathscr{A}( \mathscr{ G }_c  )[h]  \right|
	} 
	\text{ .}
	\]
	Then,
	\begin{enumerate}[label=(\roman*)\;]
		\item there is at least one non-sequential interdimensional edge $ e \in \mathscr{ E }( \mathscr{ G }_c ) $;
		
		\item $ \mathscr{ G }_c $ is \emph{not} an algorithmically snapshot-like multidimensional network.
	\end{enumerate}
	
	\begin{proof}
		The proofs of (i) and (ii) directly follow from Equations~\eqref{eqSnapshotnetwork} or~\eqref{eqMultiplex}
		where the complete proof of Theorem~\ref{thmSimplespatialTVG} (from which these equations are obtained) can be found in Appendix~\ref{sectionProofthmSimplespatialTVG}.
		Also notice that Definition~\ref{defACS} is satisfiable via Theorem~\ref{thmSimplespatialTVG}.
	\end{proof}
\end{theorem} 

Although there may be an upper bound (or ``gradient zones'') for algorithmic information content in which the sequential dimensional structure does not correspond to the sequential information growth of a network, Theorem~\ref{thmSequentialstructure} basically says that this ceases to be the case as the algorithmic information content of the network topology gets sufficiently high.
This result demonstrates what may be considered an underlying intuition behind a fundamental mathematical phenomenon.
One knows from basic properties of AIT that the increase of algorithmic randomness corresponds to a lesser amount of formal mathematical theories being able to capture, generate, or identify structural patterns in the network topology.
Thus, those formal theories (or algorithms in general) able to for example generate networks with non-sequential interdimensional edges from algorithmically snapshot-like multidimensional network (and vice-versa) will be eventually accounted for, or excluded from such an amount of theories.

In this direction, one can employ the incompressibility of multidimensional networks (see Appendix~\ref{sectionBackground}) to demonstrate network topological properties that are analogous to those in random graphs.
However, instead of these properties being obtained with probability arbitrarily close to $1$ as in statistics-based proofs when the network size increases indefinitely, they are demonstrated to happen in fixed objects or structures from phase transitions in the asymptotic dominance as the network size gets sufficiently large. 
Then, we demonstrate in Appendix~\ref{sectionTopologicalTVG} that these network topological properties in turn imply a non-sequential dimensional structure.
This reveals that information dynamics analysis in (multidimensional) networks can be used to investigate certain topological properties not only for graphs or monoplex networks, but also to investigate the structure or features of the multidimensional space (e.g., time) that a multidimensional network can be embedded into.

Previous work in~\cite{Abrahao2020cAIDistortionsCN,Abrahao2021AIDistortionsEntropy} has demonstrated that isomorphic transformations do not preserve algorithmic information.
We argue that the investigation of the extent to which those algorithmic complexity distortions (due to isomorphisms) are able to algorithmically characterize the sequential interdimensional structure such as time of complex networks is an open problem interesting and fruitful to be pursued in \emph{future research}.

\section{Conclusions}\label{sectionConclusions}

In this work, we have studied the plain and prefix algorithmic randomness of multidimensional networks that can be formally represented by multiaspect graphs (MAGs). 
We have dealt with the lossless incompressibility of multidimensional networks, especially node-aligned undirected multilayer networks or dynamic networks.

First, we have compared time-varying graphs with other snapshot-like representations of dynamic networks. 
We have shown that incompressible snapshot-dynamic networks carry an amount of topological algorithmic information (i.e., the irreducible information necessary to computationally determine the graph-theoretic representation of the respective network) that is linearly dominated by the size of the set of time instants. 
To this end, we have applied a study of a worst-case lossless compression of the algorithmically characteristic string of the network using the theoretical tools from algorithmic information theory.

Then, after an extensive analysis of previous nomenclature and assumptions in the literature, we have shown that the same results can be applied to multiplex networks, where the set of layers plays the role of the set of time instants instead.
In this regard, we have shown that both snapshot-dynamic networks and multiplex networks are two particular and distinct cases of (algorithmically) snapshot-like multidimensional networks, but that are equivalent (except by a constant that only depends on the chosen universal programming language) in terms of algorithmic information.
In addition, we have shown that 
the maximum amount of topological information of an incompressible snapshot-like multidimensional network is much smaller than the amount of topological information of an incompressible multidimensional network.

We have also investigated some topological properties of incompressible multidimensional networks. 
To this end, we have applied previous results for incompressible MAGs. 
We have shown that these networks have very short diameter, high k-connectivity, and degrees on the order of half of the network size within a strong-asymptotically dominated standard deviation. 
Therefore, these theoretical findings relate lossless compression of multidimensional networks with their network topological properties. 
For example, these properties are expected to happen in both artificial or real-world multidimensional networks with a network topology that carries a maximal and irreducible information content.
In this way, such theoretical results may give rise to future tools that could be applied to, for instance, network summarization algorithms and the reducibility problem (i.e., the problem of finding the aggregate graph that represents the original multidimensional network and preserves its core properties during network analysis). 

Furthermore, we have also shown the presence of transtemporal or crosslayer edges (i.e., edges linking vertices at non-sequential time instants or layers) in those incompressible multidimensional networks. 
Although representations of these general forms of multidimensional networks may carry much more topological information than that of snapshot-like multidimensional networks, this presence of transtemporal or crosslayer edges may not correspond to some underlying structures of real-world networks. 
Specifically, this is the case of snapshot-dynamic networks, where transtemporal edges would not have any physical correspondence to connections that `jump across time instants'.
On the other hand, in the crosslayer case, it can make sense for some multilayer networks.
Thus, with the purpose of bringing algorithmic randomness to the context of multidimensional networks, our theoretical results suggest that estimating or analyzing both the incompressibility and the network topological properties of real-world networks cannot be taken into a universal approach without previously taking into account underlying topological constraints. 
Similarly to what we have shown for snapshot-based representations of multidimensional networks, the algorithmic randomness of certain networks may be strongly dependent on the underlying constraints in the structure of the network;
and, in turn, as we have shown, the incompressibility (i.e., algorithmic randomness) of multidimensional networks implies underlying constraints in the structure of the network. 

Our results highlight a theoretical dissonance between information and structural representation of the sequential interdimensional structure of a network:
the characterization in terms of the mathematical structure representation does not always corresponds to the characterization in terms of information.
We conclude that future research on such a dissonance between information and structural representation in mathematics is crucial to formalize and allow a mathematical analysis capable detecting sequential dimensions, like time in physics, in a manner that is agnostic to human cognitive biases toward certain types of representations, theories, or algorithms.

This study and approach is also key to move forward from statistical approaches to non-statistical challenges in the context of network science, data science, and beyond, such as those related to model generation, feature selection, data dimensionality reduction, and summarization. 
For example, in dynamic multilayer networks, one important challenge is to disentangle the cause and effect between and inside different networks over time, a problem relevant to almost every area of science where processes can be represented as networks and interactions as connections to other networks.

\section{Acknowledgments}

Felipe S. Abrah\~{a}o acknowledges support from the S\~{a}o Paulo Research Foundation (FAPESP), grants $2021$/$14501$-$8$ and $2023$/$05593$-$1$.
Authors acknowledge the partial support from CNPq through their individual grants: F. S. Abrahão (313.043/2016-7), K. Wehmuth (312599/2016-1), and A. Ziviani (308.729/2015-3). Authors acknowledge the INCT in Data Science – INCT-CiD (CNPq 465.560/2014-8). Authors also acknowledge the partial support from CAPES/STIC-AmSud (18-STIC-07), FAPESP (2015/24493-1), and FAPERJ (E-26/203.046/2017). H. Zenil acknowledges the Swedish Research Council (Vetenskapsr\r{a}det) for their support under grant No. 2015-05299.

\bibliographystyle{plain}
\bibliography{2.2.1-CompleteRefs-Felipe.bib}

\begin{bibdiv}
\begin{biblist}

\bib{Abrahao2023bSemanticrobustnessCLPMST}{inproceedings}{
      author={Abrah{\~{a}}o, Felipe~S.},
      author={D'Ottaviano, Itala M.~L.},
       title={Semantic robustness to isomorphisms in partial dyadic
  structures},
        date={2023},
   booktitle={17th international congress on logic methodology and philosophy
  of science and technology ({CLMPST})},
        note={Available at:
  \url{http://dx.doi.org/10.13140/RG.2.2.34087.62880}},
}

\bib{Abrahao2020cAIDistortionsCN}{inproceedings}{
      author={Abrah{\~{a}}o, Felipe~S.},
      author={Wehmuth, Klaus},
      author={Zenil, Hector},
      author={Ziviani, Artur},
       title={An {Algorithmic} {Information} {Distortion} in {Multidimensional}
  {Networks}},
        date={2021},
   booktitle={Complex {Networks} \& {Their} {Applications} {IX}},
      editor={Benito, Rosa~M.},
      editor={Cherifi, Chantal},
      editor={Cherifi, Hocine},
      editor={Moro, Esteban},
      editor={Rocha, Luis~Mateus},
      editor={Sales-Pardo, Marta},
      series={Studies in Computational Intelligence},
      volume={944},
   publisher={Springer International Publishing},
     address={Cham},
       pages={520\ndash 531},
        note={Available at: \url{http://doi.org/10.1007/978-3-030-65351-4_42}},
}

\bib{Abrahao2021AIDistortionsEntropy}{article}{
      author={Abrah{\~{a}}o, Felipe~S.},
      author={Wehmuth, Klaus},
      author={Zenil, Hector},
      author={Ziviani, Artur},
       title={Algorithmic information distortions in node-aligned and
  node-unaligned multidimensional networks},
        date={2021},
        ISSN={1099-4300},
     journal={Entropy},
      volume={23},
      number={7},
      eprint={http://doi.org/10.3390/e23070835},
}

\bib{Abrahao2018darxivandreport}{article}{
      author={Abrah{\~{a}}o, Felipe~S.},
      author={Wehmuth, Klaus},
      author={Zenil, Hector},
      author={Ziviani, Artur},
       title={{Algorithmic information distortions and incompressibility in
  uniform multidimensional networks}},
        date={2023},
     journal={arXiv Preprints},
        note={Preprint available at: \url{https://arxiv.org/abs/1810.11719}.
  Also available as the research report no. 8/2018, National Laboratory for
  Scientific Computing (LNCC), Petr{\'{o}}polis, Brazil.},
}

\bib{Abrahao2018publishedAMS}{article}{
      author={Abrah{\~{a}}o, Felipe~S.},
      author={Wehmuth, Klaus},
      author={Ziviani, Artur},
       title={{Emergent Open-Endedness from Contagion of the Fittest}},
        date={2018},
     journal={Complex Systems},
      volume={27},
      number={04},
      eprint={{http://doi.org/10.25088/ComplexSystems.27.4.369}},
}

\bib{Abrahao2017published}{article}{
      author={Abrah{\~{a}}o, Felipe~S.},
      author={Wehmuth, Klaus},
      author={Ziviani, Artur},
       title={{Algorithmic networks: Central time to trigger expected emergent
  open-endedness}},
        date={2019sep},
        ISSN={03043975},
     journal={Theoretical Computer Science},
      volume={785},
       pages={83\ndash 116},
      eprint={https://doi.org/10.1016/j.tcs.2019.03.008},
}

\bib{ABRAHAO2019BAMS}{inproceedings}{
      author={Abrah{\~{a}}o, Felipe~S.},
      author={Wehmuth, Klaus},
      author={Ziviani, Artur},
       title={{Transtemporal edges and crosslayer edges in incompressible
  high-order networks}},
        date={2019},
   booktitle={{Meeting on Theory of Computation (ETC)}, {Congress of the
  Brazilian Computer Society (CSBC)}},
   publisher={Brazilian Computer Society (SBC)},
     address={Bel{\'{e}}m},
        note={Available at
  \url{https://sol.sbc.org.br/index.php/etc/article/view/6389}},
}

\bib{Barabasi2009a}{article}{
      author={Barab{\'{a}}si, Albert-L{\'{a}}szl{\'{o}}},
       title={{Scale-Free Networks: A Decade and Beyond}},
        date={2009jul},
        ISSN={0036-8075},
     journal={Science},
      volume={325},
      number={5939},
       pages={412\ndash 413},
      eprint={http://www.sciencemag.org/cgi/doi/10.1126/science.1173299},
}

\bib{Barmpalias2019}{article}{
      author={Barmpalias, George},
      author={Lewis-Pye, Andrew},
       title={{Compression of Data Streams Down to Their Information Content}},
        date={2019jul},
        ISSN={0018-9448},
     journal={IEEE Transactions on Information Theory},
      volume={65},
      number={7},
       pages={4471\ndash 4485},
      eprint={http://doi.org/10.1109/TIT.2019.2896638},
}

\bib{Becher2002}{article}{
      author={Becher, Ver{\'{o}}nica},
      author={Figueira, Santiago},
       title={{An example of a computable absolutely normal number}},
        date={2002jan},
        ISSN={03043975},
     journal={Theoretical Computer Science},
      volume={270},
      number={1-2},
       pages={947\ndash 958},
      eprint={{https://doi.org/10.1016/S0304-3975(01)00170-0}},
}

\bib{Becher2015}{article}{
      author={Becher, Ver{\'{o}}nica},
      author={Heiber, Pablo~A.},
      author={Slaman, Theodore~A.},
       title={{A computable absolutely normal Liouville number}},
        date={2015},
     journal={Mathematics of Computation},
      volume={84},
       pages={2939\ndash 2952},
      eprint={https://doi.org/10.1090/mcom/2964},
}

\bib{Becher2013}{article}{
      author={Becher, Ver{\'{o}}nica},
      author={Heiber, Pablo~Ariel},
      author={Slaman, Theodore~A.},
       title={{A polynomial-time algorithm for computing absolutely normal
  numbers}},
        date={2013nov},
        ISSN={08905401},
     journal={Information and Computation},
      volume={232},
       pages={1\ndash 9},
      eprint={https://doi.org/10.1016/j.ic.2013.08.013},
}

\bib{Berlingerio2011}{inproceedings}{
      author={Berlingerio, Michele},
      author={Coscia, Michele},
      author={Giannotti, Fosca},
      author={Monreale, Anna},
      author={Pedreschi, Dino},
       title={{Foundations of Multidimensional Network Analysis}},
        date={2011jul},
   booktitle={2011 international conference on advances in social networks
  analysis and mining},
   publisher={IEEE},
       pages={485\ndash 489},
}

\bib{Boccaletti2014a}{article}{
      author={Boccaletti, S.},
      author={Bianconi, G.},
      author={Criado, R.},
      author={del Genio, C.I.},
      author={G{\'{o}}mez-Garde{\~{n}}es, J.},
      author={Romance, M.},
      author={Sendi{\~{n}}a-Nadal, I.},
      author={Wang, Z.},
      author={Zanin, M.},
       title={{The structure and dynamics of multilayer networks}},
        date={2014nov},
        ISSN={03701573},
     journal={Physics Reports},
      volume={544},
      number={1},
       pages={1\ndash 122},
      eprint={https://linkinghub.elsevier.com/retrieve/pii/S0370157314002105},
}

\bib{Bollobas1998}{book}{
      author={Bollob{\'{a}}s, B{\'{e}}la},
       title={{Modern graph theory}},
      series={Graduate texts in mathematics},
   publisher={Springer Science {\&} Business Media},
        date={1998},
        ISBN={0387984887},
}

\bib{Brandes2005a}{incollection}{
      author={Brandes, Ulrik},
      author={Erlebach, Thomas},
       title={{Fundamentals}},
   booktitle={Network analysis},
      editor={Brandes, Ulrik},
      editor={Erlebach, Thomas},
     address={Berlin, Heidelberg},
}

\bib{Buhrman1999}{article}{
      author={Buhrman, Harry},
      author={Li, Ming},
      author={Tromp, John},
      author={Vit{\'{a}}nyi, Paul},
       title={{Kolmogorov Random Graphs and the Incompressibility Method}},
        date={1999jan},
        ISSN={0097-5397},
     journal={SIAM Journal on Computing},
      volume={29},
      number={2},
       pages={590\ndash 599},
      eprint={http://epubs.siam.org/doi/10.1137/S0097539797327805},
}

\bib{Calude2002}{book}{
      author={Calude, Cristian~S.},
       title={{Information and Randomness: An algorithmic perspective}},
     edition={2},
   publisher={Springer-Verlag},
        date={2002},
        ISBN={3540434666},
}

\bib{Calude2009}{incollection}{
      author={Calude, Cristian~S.},
       title={{Information: The Algorithmic Paradigm}},
        date={2009},
   booktitle={Formal theories of information},
      editor={Sommaruga, Giovanni},
      series={Lecture Notes in Computer Science},
      volume={5363 LNCS},
       pages={79\ndash 94},
}

\bib{Chaitin2004}{book}{
      author={Chaitin, Gregory},
       title={{Algorithmic Information Theory}},
     edition={3},
   publisher={Cambridge University Press},
        date={2004},
        ISBN={0521616042},
}

\bib{Contractor2011}{article}{
      author={Contractor, Noshir~S.},
      author={Monge, Peter~R.},
      author={Leonardi, Paul~M.},
       title={{Multidimensional networks and the dynamics of sociomateriality:
  Bringing technology inside the network}},
        date={2011},
     journal={International Journal of Communication},
      volume={5},
       pages={682\ndash 720},
         url={https://ijoc.org/index.php/ijoc/article/view/1131},
}

\bib{Costa2015a}{article}{
      author={Costa, Eduardo~Chinelate},
      author={Vieira, Alex~Borges},
      author={Wehmuth, Klaus},
      author={Ziviani, Artur},
      author={da~Silva, Ana Paula~Couto},
       title={{Time Centrality in Dynamic Complex Networks}},
        date={2015},
        ISSN={02195259},
     journal={Advances in Complex Systems},
      volume={18},
      number={07n08},
      eprint={http://dx.doi.org/10.1142/S021952591550023X},
}

\bib{Cover2005}{book}{
      author={Cover, Thomas~M.},
      author={Thomas, Joy~A.},
       title={{Elements of Information Theory}},
   publisher={John Wiley {\&} Sons, Inc.},
     address={Hoboken, NJ, USA},
        date={2005},
        ISBN={9780471241959},
}

\bib{Cozzo2018}{book}{
      author={Cozzo, Emanuele},
      author={de~Arruda, Guilherme~Ferraz},
      author={Rodrigues, Francisco~Aparecido},
      author={Moreno, Yamir},
       title={{Multiplex Networks}},
      series={SpringerBriefs in Complexity},
   publisher={Springer International Publishing},
     address={Cham},
        date={2018},
        ISBN={978-3-319-92254-6},
}

\bib{Domenico2016}{article}{
      author={{De Domenico}, Manlio},
      author={Granell, Clara},
      author={Porter, Mason~A.},
      author={Arenas, Alex},
       title={{The physics of spreading processes in multilayer networks}},
        date={2016oct},
        ISSN={1745-2473},
     journal={Nature Physics},
      volume={12},
      number={10},
       pages={901\ndash 906},
      eprint={http://www.nature.com/articles/nphys3865},
}

\bib{Domenico2015}{article}{
      author={{De Domenico}, Manlio},
      author={Nicosia, Vincenzo},
      author={Arenas, Alexandre},
      author={Latora, Vito},
       title={{Structural reducibility of multilayer networks}},
        date={2015dec},
        ISSN={2041-1723},
     journal={Nature Communications},
      volume={6},
      number={1},
       pages={6864},
      eprint={http://www.nature.com/articles/ncomms7864},
}

\bib{Domenico2013}{article}{
      author={{De Domenico}, Manlio},
      author={Sol{\'{e}}-Ribalta, Albert},
      author={Cozzo, Emanuele},
      author={Kivel{\"{a}}, Mikko},
      author={Moreno, Yamir},
      author={Porter, Mason~A.},
      author={G{\'{o}}mez, Sergio},
      author={Arenas, Alex},
       title={{Mathematical Formulation of Multilayer Networks}},
        date={2013dec},
        ISSN={2160-3308},
     journal={Physical Review X},
      volume={3},
      number={4},
       pages={041022},
      eprint={https://link.aps.org/doi/10.1103/PhysRevX.3.041022},
}

\bib{Dehmer2011}{book}{
      editor={Dehmer, Matthias},
      editor={Emmert-Streib, Frank},
      editor={Mehler, Alexander},
       title={{Towards an Information Theory of Complex Networks}},
   publisher={Birkh{\"{a}}user Boston},
     address={Boston},
        date={2011},
        ISBN={978-0-8176-4903-6},
}

\bib{Delahaye2012}{article}{
      author={Delahaye, Jean~Paul},
      author={Zenil, Hector},
       title={{Numerical evaluation of algorithmic complexity for short
  strings: A glance into the innermost structure of randomness}},
        date={2012sep},
        ISSN={00963003},
     journal={Applied Mathematics and Computation},
      volume={219},
      number={1},
       pages={63\ndash 77},
      eprint={http://doi.org/10.1016/j.amc.2011.10.006},
}

\bib{Diestel2017}{book}{
      author={Diestel, Reinhard},
       title={{Graph Theory}},
    language={English},
     edition={5},
      series={Graduate Texts in Mathematics},
   publisher={Springer-Verlag Berlin Heidelberg},
        date={2017},
      volume={173},
        ISBN={978-3-662-53621-6},
         url={https://www.springer.com/br/book/9783662536216},
}

\bib{Downey2010}{book}{
      author={Downey, Rodney~G.},
      author={Hirschfeldt, Denis~R.},
       title={{Algorithmic Randomness and Complexity}},
      series={Theory and Applications of Computability},
   publisher={Springer New York},
     address={New York, NY},
        date={2010},
        ISBN={978-0-387-95567-4},
}

\bib{Harrison-Trainor2019}{article}{
      author={Harrison-Trainor, Matthew},
      author={Khoussainov, Bakh},
      author={Turetsky, Daniel},
       title={{Effective aspects of algorithmically random structures}},
        date={2019sep},
        ISSN={22113576},
     journal={Computability},
      volume={8},
      number={3-4},
       pages={359\ndash 375},
      eprint={https://doi.org/10.3233/COM-180101},
}

\bib{Khoussainov2014}{inproceedings}{
      author={Khoussainov, Bakhadyr},
       title={{A quest for algorithmically random infinite structures}},
        date={2014},
   booktitle={{Proceedings of the Joint Meeting of the Twenty-Third EACSL
  Annual Conference on Computer Science Logic (CSL) and the Twenty-Ninth Annual
  ACM/IEEE Symposium on Logic in Computer Science (LICS) - CSL-LICS}},
   publisher={ACM Press},
     address={New York, New York, USA},
       pages={1\ndash 9},
}

\bib{Kivela2014}{article}{
      author={Kivela, M.},
      author={Arenas, Alex},
      author={Barthelemy, Marc},
      author={Gleeson, James~P.},
      author={Moreno, Yamir},
      author={Porter, Mason~A.},
       title={{Multilayer networks}},
        date={2014sep},
        ISSN={2051-1310},
     journal={Journal of Complex Networks},
      volume={2},
      number={3},
       pages={203\ndash 271},
  eprint={https://academic.oup.com/comnet/article-lookup/doi/10.1093/comnet/cnu016},
}

\bib{Lambiotte2019}{article}{
      author={Lambiotte, Renaud},
      author={Rosvall, Martin},
      author={Scholtes, Ingo},
       title={{From networks to optimal higher-order models of complex
  systems}},
        date={2019apr},
        ISSN={1745-2481},
     journal={Nature Physics},
      volume={15},
      number={4},
       pages={313\ndash 320},
      eprint={http://www.nature.com/articles/s41567-019-0459-y},
}

\bib{Leung-Yan-Cheong1978}{article}{
      author={Leung-Yan-Cheong, S.},
      author={Cover, T.},
       title={{Some equivalences between Shannon entropy and Kolmogorov
  complexity}},
        date={1978may},
        ISSN={0018-9448},
     journal={IEEE Transactions on Information Theory},
      volume={24},
      number={3},
       pages={331\ndash 338},
      eprint={http://doi.org/10.1109/TIT.1978.1055891},
}

\bib{Lewis2009}{book}{
      author={Lewis, Ted~G.},
       title={{Network Science}},
   publisher={John Wiley {\&} Sons, Inc.},
     address={Hoboken, NJ, USA},
        date={2009},
        ISBN={9780470400791},
}

\bib{Li1997}{book}{
      author={Li, Ming},
      author={Vit{\'{a}}nyi, Paul},
       title={{An Introduction to Kolmogorov Complexity and Its Applications}},
     edition={4},
      series={Texts in Computer Science},
     address={Cham},
        date={2019},
        ISBN={978-3-030-11298-1},
}

\bib{Michail2015}{incollection}{
      author={Michail, Othon},
       title={{An Introduction to Temporal Graphs: An Algorithmic
  Perspective}},
        date={2015},
   booktitle={Algorithms, probability, networks, and games},
      editor={Zaroliagis, Christos},
      editor={Pantziou, Grammati},
      editor={Kontogiannis, Spyros},
      series={Lecture Notes in Computer Science},
      volume={9295},
   publisher={Springer International Publishing},
     address={Cham},
       pages={308\ndash 343},
}

\bib{Michail2018}{article}{
      author={Michail, Othon},
      author={Spirakis, Paul~G.},
       title={{Elements of the theory of dynamic networks}},
        date={2018jan},
        ISSN={00010782},
     journal={Communications of the ACM},
      volume={61},
      number={2},
       pages={72},
      eprint={http://dl.acm.org/citation.cfm?doid=3181977.3156693},
}

\bib{Morzy2017a}{article}{
      author={Morzy, Miko{\l}aj},
      author={Kajdanowicz, Tomasz},
      author={Kazienko, Przemys{\l}aw},
       title={{On Measuring the Complexity of Networks: Kolmogorov Complexity
  versus Entropy}},
        date={2017},
        ISSN={1076-2787},
     journal={Complexity},
      volume={2017},
       pages={1\ndash 12},
      eprint={https://doi.org/10.1155/2017/3250301},
}

\bib{Mowshowitz2012}{article}{
      author={Mowshowitz, Abbe},
      author={Dehmer, Matthias},
       title={{Entropy and the complexity of graphs revisited}},
        date={2012mar},
        ISSN={10994300},
     journal={Entropy},
      volume={14},
      number={3},
       pages={559\ndash 570},
      eprint={http://www.mdpi.com/1099-4300/14/3/559/},
}

\bib{Pan2011}{article}{
      author={Pan, Raj~Kumar},
      author={Saram{\"{a}}ki, Jari},
       title={{Path lengths, correlations, and centrality in temporal
  networks}},
        date={2011jul},
        ISSN={1539-3755},
     journal={Physical Review E},
      volume={84},
      number={1},
       pages={016105},
      eprint={https://link.aps.org/doi/10.1103/PhysRevE.84.016105},
}

\bib{Pilosof2017}{article}{
      author={Pilosof, Shai},
      author={Porter, Mason~A.},
      author={Pascual, Mercedes},
      author={K{\'{e}}fi, Sonia},
       title={{The multilayer nature of ecological networks}},
        date={2017mar},
        ISSN={2397-334X},
     journal={Nature Ecology {\&} Evolution 2017},
      volume={1},
      number={4},
       pages={0101},
      eprint={https://www.nature.com/articles/s41559-017-0101},
}

\bib{Radicchi2017}{article}{
      author={Radicchi, Filippo},
      author={Bianconi, Ginestra},
       title={{Redundant Interdependencies Boost the Robustness of Multiplex
  Networks}},
        date={2017jan},
        ISSN={2160-3308},
     journal={Physical Review X},
      volume={7},
      number={1},
       pages={011013},
      eprint={https://link.aps.org/doi/10.1103/PhysRevX.7.011013},
}

\bib{Rossetti2018b}{article}{
      author={Rossetti, Giulio},
      author={Cazabet, R{\'{e}}my},
       title={{Community Discovery in Dynamic Networks}},
        date={2018feb},
        ISSN={03600300},
     journal={ACM Computing Surveys},
      volume={51},
      number={2},
       pages={1\ndash 37},
      eprint={http://dl.acm.org/citation.cfm?doid=3186333.3172867},
}

\bib{Sow2003}{article}{
      author={Sow, D.M.},
      author={Eleftheriadis, A.},
       title={{Complexity distortion theory}},
        date={2003mar},
        ISSN={0018-9448},
     journal={IEEE Transactions on Information Theory},
      volume={49},
      number={3},
       pages={604\ndash 608},
      eprint={http://doi.org/10.1109/TIT.2002.808135},
}

\bib{Wehmuth2018a}{inproceedings}{
      author={Wehmuth, Klaus},
      author={Costa, Bernardo},
      author={Bechara, Jo{\~{a}}o~Victor},
      author={Ziviani, Artur},
       title={{A Multilayer and Time-Varying Structural Analysis of the
  Brazilian Air Transportation Network}},
        date={2018},
   booktitle={Latin america data science workshop (ladas) co-located with 44th
  international conference on very large data bases (vldb 2018)},
       pages={57\ndash 64},
}

\bib{Wehmuth2016b}{article}{
      author={Wehmuth, Klaus},
      author={Fleury, {\'{E}}ric},
      author={Ziviani, Artur},
       title={{On MultiAspect graphs}},
        date={2016},
        ISSN={03043975},
     journal={Theoretical Computer Science},
      volume={651},
       pages={50\ndash 61},
      eprint={http://doi.org/10.1016/j.tcs.2016.08.017},
}

\bib{Wehmuth2017}{article}{
      author={Wehmuth, Klaus},
      author={Fleury, {\'{E}}ric},
      author={Ziviani, Artur},
       title={{MultiAspect Graphs: Algebraic Representation and Algorithms}},
        date={2017},
        ISSN={1999-4893},
     journal={Algorithms},
      volume={10},
      number={1},
       pages={1\ndash 36},
      eprint={http://www.mdpi.com/1999-4893/10/1/1},
}

\bib{Wehmuth2018}{inproceedings}{
      author={Wehmuth, Klaus},
      author={Ziviani, Artur},
       title={{Avoiding Spurious Paths in Centralities Based on Shortest Paths
  in High Order Networks}},
        date={2018oct},
   booktitle={2018 eighth latin-american symposium on dependable computing
  (ladc)},
   publisher={IEEE},
       pages={19\ndash 26},
}

\bib{Wehmuth2018b}{inproceedings}{
      author={Wehmuth, Klaus},
      author={Ziviani, Artur},
       title={{Centralities in High Order Networks}},
    language={English},
        date={2018},
   booktitle={{Meeting on Theory of Computation (ETC), Congress of the
  Brazilian Computer Society (SBC) 2018}},
      volume={3},
   publisher={Brazilian Computer Society (SBC)},
  url={http://portaldeconteudo.sbc.org.br/index.php/etc/article/view/3147},
}

\bib{Wehmuth2015a}{inproceedings}{
      author={Wehmuth, Klaus},
      author={Ziviani, Artur},
      author={Fleury, Eric},
       title={{A unifying model for representing time-varying graphs}},
        date={2015oct},
   booktitle={{Proc. of the IEEE Int. Conf. on Data Science and Advanced
  Analytics (DSAA)}},
   publisher={IEEE},
       pages={1\ndash 10},
}

\bib{Zarate2019nat}{book}{
      editor={Zarate, Jean~Mary},
      editor={Klopper, Abigail},
      editor={Levi, Federico},
      editor={Sutherland, Mary~Elizabeth},
        note={\emph{Nature}, Collection, 2019. URL \url{https:
  //www.nature.com/collections/adajhgjece}.},
}

\bib{Zenil2018}{article}{
      author={Zenil, Hector},
      author={Hern{\'{a}}ndez-Orozco, Santiago},
      author={Kiani, Narsis},
      author={Soler-Toscano, Fernando},
      author={Rueda-Toicen, Antonio},
      author={Tegn{\'{e}}r, Jesper},
       title={{A Decomposition Method for Global Evaluation of Shannon Entropy
  and Local Estimations of Algorithmic Complexity}},
        date={2018aug},
        ISSN={1099-4300},
     journal={Entropy},
      volume={20},
      number={8},
       pages={605},
      eprint={http://arxiv.org/abs/1609.00110
  http://www.mdpi.com/1099-4300/20/8/605},
}

\bib{Zenil2018a}{article}{
      author={Zenil, Hector},
      author={Kiani, Narsis},
      author={Tegn{\'{e}}r, Jesper},
       title={{A Review of Graph and Network Complexity from an Algorithmic
  Information Perspective}},
        date={2018jul},
        ISSN={1099-4300},
     journal={Entropy},
      volume={20},
      number={8},
       pages={551},
      eprint={http://www.mdpi.com/1099-4300/20/8/551},
}

\bib{Zenil2017a}{article}{
      author={Zenil, Hector},
      author={Kiani, Narsis~A.},
      author={Tegn{\'{e}}r, Jesper},
       title={{Low-algorithmic-complexity entropy-deceiving graphs}},
        date={2017jul},
        ISSN={2470-0045},
     journal={Physical Review E},
      volume={96},
      number={1},
       pages={012308},
      eprint={http://dx.doi.org/10.1103/PhysRevE.96.012308},
}

\end{biblist}
\end{bibdiv}

\appendix

\section{Preliminary definitions and notation}\label{appendixBackground}

\subsection{Multiaspect graphs}\label{appendixBackgroundMAGs}

We directly base our definitions and notation regarding classical graphs on \cite{Diestel2017,Brandes2005a,Bollobas1998}; and regarding generalized graph representation of dyadic relations between $n$-tuples (i.e. multiaspect graphs) on \cite{Wehmuth2016b,Wehmuth2017}.

\begin{definition}\label{defMAG}
	Let $ \mathscr{G}=(\mathscr{A},\mathscr{E}) $ be a multiaspect graph (MAG), where $\mathscr{E}$ is the set of existing composite edges of the MAG and $\mathscr{A}$ is a class of sets, each of which is an \emph{aspect}. Each aspect $ \mathbf{ \sigma } \in \mathscr{A} $ is a finite set and the number of aspects $ p = | \mathscr{A} | $ is called the \emph{order} of $ \mathscr{G} $. By an immediate convention, we call a MAG with only one aspect as a \emph{first order} MAG, a MAG with two aspects as a \emph{second order} MAG and so on.  Each composite edge (or arrow) $ e \in \mathscr{E} $ may be denoted by an ordered $2p$-tuple $ ( a_1,\dots,a_p, b_1, \dots, b_p ) $, where $ a_i, b_i $ are elements of the $i$-th aspect with $ 1 \leq i \leq p = | \mathscr{A} | $. 
	
\end{definition} 

Note that $ \mathscr{A}( \mathscr{G} ) $ denotes the class of aspects of $ \mathscr{G} $ and $ \mathscr{E}( \mathscr{G} ) $ denotes the \emph{composite edge set} of $ \mathscr{G} $. We denote the $i$-th aspect of $ \mathscr{G} $ as $ \mathscr{A}( \mathscr{G} )[i] $. So, $ | \mathscr{A}( \mathscr{G} )[i] |$ denotes the number of elements in $ \mathscr{A}( \mathscr{G} )[i] $. In order to match the classical graph case, we adopt the convention of calling the elements of the first aspect of a MAG as \emph{vertices}. Therefore, we also denote the set $ \mathscr{A}( \mathscr{G} )[1] $ of elements of the first aspect of a MAG  $ \mathscr{G} $ as $ V( \mathscr{G} ) $. Thus, a vertex should not be confused with a composite vertex. The set of all \emph{composite vertices} $ \mathbf{v} $ of $ \mathscr{G} $ is defined by
\[
\mathbb{V}( \mathscr{G} ) = \bigtimes_{i=1}^{p} \mathscr{A}( \mathscr{G} )[i] 
\]
\noindent  and the set of all \emph{composite edges} $ e $ of $ \mathscr{G} $ is defined by
\[
\mathbb{E}(\mathscr{G}) = \bigtimes_{n=1}^{2p}  \mathscr{A}(G)[ { (n-1)\pmod{p} } + 1 )]  \text{ ,}
\]
\noindent so that, for every ordered pair $ ( \mathbf{u} , \mathbf{v} ) $ with $ \mathbf{u} , \mathbf{v} \in \mathbb{V}( \mathscr{G} )  $, we have $ ( \mathbf{u} , \mathbf{v} ) = e \in \mathbb{E}( \mathscr{G} )  $. Also, for every $ e \in \mathbb{E}( \mathscr{G} )  $, we have $ ( \mathbf{u} , \mathbf{v} ) = e $ such that $ \mathbf{u} , \mathbf{v} \in \mathbb{V}( \mathscr{G} )  $. Thus,
\[
\mathscr{E}( \mathscr{G} ) \subseteq \mathbb{E}(\mathscr{G})
\] 

\begin{definition}\label{defTraditionalMAG}
	We say a \emph{directed} MAG (or graph) without \emph{self-loops} is a \emph{traditional} directed MAG (or graph), denoted as $  \mathscr{G}_d = (\mathscr{A},\mathscr{E}) $.
\end{definition}

\begin{definition}\label{defSimplifiedMAG}
	We say an \emph{undirected} MAG (or graph)  without \emph{self-loops} is a \emph{simple} MAG (or graph), denoted as $  \mathscr{G}_c = (\mathscr{A},\mathscr{E}) $,
	so that the set of all possible composite edges $ \mathbb{E} $ is subjected to a restriction in the form
	\[
	\mathscr{E}( \mathscr{G}_c )  \subseteq \mathbb{E}_c( \mathscr{G}_c ) \coloneqq   \{ \{ \mathbf{u} , \mathbf{v} \} \mid \mathbf{u},\mathbf{v} \in \mathbb{V}( \mathscr{G}_c )  \} 
	\text{ ,}
	\]
	\noindent where there is $ Y \subseteq \mathbb{E}(\mathscr{G}_c)  $ such that
	\[
	\{ \mathbf{u} , \mathbf{v} \} \in \mathscr{E}( \mathscr{G}_c )  \iff 
	( \mathbf{u} , \mathbf{v} ) \in Y \, \land \, ( \mathbf{v} , \mathbf{u}) \in Y \land \mathbf{u} \neq \mathbf{v}
	\]
\end{definition}

We have directly from this Definition~\ref{defSimplifiedMAG} that
\[
\left| \mathbb{E}_c( \mathscr{G}_c )   \right| = \frac{ { \left| \mathbb{V}( \mathscr{G}_c ) \right| }^2 - { \left| \mathbb{V}( \mathscr{G}_c ) \right| } }{ 2 }
\text{ .}
\]

Concerning the presence or absence of composite edges in $ \mathscr{E}( \mathscr{G}_c  ) $, we defined the characteristic string \cite{Abrahao2018darxivandreport} of a simple MAG by previously fixing an arbitrary ordering of all possible composite edges.
However, this condition becomes unnecessary for recursively labeled families, since this ordering is already embedded in Definition~\ref{defLabeledfamilyofMAG}.
See also Lemma~\ref{lemmaLabeledfamilyofMAG}.

\begin{definition}\label{BdefCharacteristicstringofasimpleMAG}
	Let $ \left( e_1 , \dots , e_{ \left| \mathbb{E}_c( \mathscr{G}_c )   \right| } \right) $ be any arbitrarily fixed ordering of all possible composite edges of a simple MAG $ \mathscr{G}_c $.
	We say that a string $ x \in \{ 0 , 1 \}^* $ with $ l( x ) = \left| \mathbb{E}_c( \mathscr{G}_c )   \right| $ is a \emph{characteristic string} of a simple MAG $ \mathscr{G}_c $
	\textit{iff}, for every $ e_j \in \mathbb{E}_c( \mathscr{G}_c ) $,
	\[
	e_j  \in \mathscr{E}( \mathscr{G}_c  ) \iff \text{ the $j$-th digit in $x$ is $1$}
	\text{ ,}
	\]
	\noindent  where $ 1 \leq j \leq l( x ) $. 
	
\end{definition}

We define the \emph{composite diameter} of $  \mathscr{G} $ in an analogous way to diameter in classical graphs:

\begin{definition}\label{BdefCompositediameter}
	The composite diameter $ \mathrm{D}_\mathscr{E}( \mathscr{G} ) $ is the maximum value in the set of the minimum number of steps (through composite edges) in $ \mathscr{E}( \mathscr{G} ) $ necessary to reach a composite vertex $ \mathbf{v}$ from a composite vertex $ \mathbf{u} $, for any $  \mathbf{u} , \mathbf{v} \in \mathbb{V}( \mathscr{G} ) $.

\end{definition}

See also \cite{Wehmuth2016b} for paths and distances in MAGs. 
Moreover, as in \cite{Wehmuth2016b}:

\begin{definition}\label{defMAGandgraphsisomorphisms}
	We say a traditional MAG $ \mathscr{G}_d $  is \emph{isomorphic} to a traditional directed graph $ G $ when there is a bijective function $ f : \mathbb{V}( \mathscr{G}_d ) \to V( G ) $ such that an edge $ e \in \mathscr{E}( \mathscr{G}_d ) $ if, and only if, the edge $ ( f( \pi_o( e ) ) , f( \pi_d( e ) ) ) \in E( G )$, where $ \pi_o $ is a function that returns the origin vertex of an edge and the function $ \pi_d $ is a function that returns the destination vertex of an edge. 
\end{definition}

Thus, the reader may notice that the aspects in $\mathscr{A}$ determine how the set $\mathscr{E}$ will be defined and, therefore, they determine the type of network that the MAG is univocally representing: for example, a time-varying graph (TVG) as in~\cite{Costa2015a,Wehmuth2015a} or a multilayer graph as in~\cite{Wehmuth2018a}. 
In the particular case of dynamic networks, as defined in \cite{Costa2015a,Wehmuth2015a}, we have that:

\begin{definition}\label{BdefTVG} 
	Let $\mathrm{ G_t }=(\mathrm{V},\mathscr{E},\mathrm{T})$ be a second order MAG representing a \emph{time-varying graph} (TVG), where $\mathrm{V}$ is the set of vertices, $\mathrm{T}$ is the set of time instants, and $\mathscr{E} \subseteq \mathrm{V} \times \mathrm{T} \times \mathrm{V} \times \mathrm{T}$ is the set of (composite) edges. 
	We denote the set of time instants of the graph $\mathrm{ G_t }=(\mathrm{V},\mathscr{E},\mathrm{T})$ by $ \mathrm{T}(\mathrm{ G_t })=\{t_0, t_1, \dotsc, t_{|\mathrm{T}(\mathrm{ G_t })|-1} \} $. 
	Also, let $\mathrm{V}(\mathrm{ G_t })$ denote the set of vertices of $\mathrm{ G_t }$ and $|\mathrm{V}(\mathrm{ G_t })|$ denote the cardinality of the set of vertices in $\mathrm{ G_t }$.
\end{definition}

Immediately, one indeed has from Definition~\ref{defMAG} \cite{Wehmuth2016b,Wehmuth2017} that classical graphs are particular cases of MAGs, which have only one aspect:

\begin{definition}\label{defClassicgraph}
	A labeled \emph{undirected} graph $ G = ( V , E )  $ without \emph{self-loops} is a labeled graph with a restriction $ \mathbb{E}_c $ in the edge set $E$ such that each edge is an \emph{unordered pair} with
	\[
	E \subseteq \mathbb{E}_c\left(  G \right) \coloneq \{ \{ x , y \} \mid x , y \in V \}  
	\] 
	\noindent where\footnote{ That is, the adjacency matrix of this graph is symmetric and the main diagonal is null. } there is $ Y \subseteq  V \times V $ such that
	\[
	\{ x , y \} \in E \subseteq \mathbb{E}_c\left(  G \right) \iff ( x , y ) \in Y \, \land \, ( y , x ) \in Y \, \land \, x \neq y
	\]
	\noindent We also refer to these graphs as \emph{classical} (or \emph{simple labeled}) graphs.
\end{definition}

The terms \emph{vertex} and \emph{node} may be employed interchangeably in this article. 
However, note that we rather choose to use the term \emph{node} preferentially within the context of networks, where nodes may realize operations, computations or would have some kind of agency, like in real-world networks. 
Thus, we rather choose to use the term \emph{vertex} preferentially in the mathematical context of graph theory.

\subsection{Algorithmic information theory}\label{appendixBackgroundAIT}

From \cite{Li1997,Downey2010,Chaitin2004,Calude2002}, we will briefly recover in this section some of the main definitions, concepts, and notation in \emph{algorithmic information} theory (aka Kolmogorov complexity theory or Solomonoff-Kolmogorov-Chaitin complexity theory).
Let $ \{ 0 , 1 \}^* $ be the set of all finite binary strings.
Let $ l(x) $ denote the length of a finite string $ x \in  \{ 0 , 1 \}^* $.
In addition, let $ | X | $ denote the number of elements (i.e., the cardinality) in a set, if $ X $ is a finite set.
Let $ { ( x ) }_2 $ denote the string which is a binary represenation of the number $ x $. In addition, let $ (x)_{L} $ denote the representation of the number $ x \in \mathbb{N} $ in language $ L $.
Let $ \mathbf{L_U} $ denote a binary universal programming language for a universal Turing machine $\mathbf{U}$. 
Let $ \mathbf{L'_U} $ denote a binary \emph{prefix-free} (or \emph{self-delimiting}) universal programming language for a prefix universal Turing machine $\mathbf{U}$.
Let $ \left< \, \cdot \, , \, \cdot \, \right> $ denote an arbitrary recursive bijective pairing function. This notation can be recursively extended to $ \left<   \, \cdot \, ,  \left< \, \cdot \, , \, \cdot \, \right> \right> $ and, then, to an ordered tuple $ \left< \, \cdot \, , \, \cdot \,  \, , \, \cdot \,\right> $. This iteration can be recursively applied with the purpose of defining finite ordered tuples $ \left< \cdot \, , \, \dots \, , \, \cdot   \right> $.

\begin{definition}\label{BdefC}
	The (unconditional) \emph{plain} \emph{algorithmic complexity} (also known as C-complexity, plain Kolmogorov complexity, plain program-size complexity or plain Solomonoff-Komogorov-Chaitin complexity) of a finite binary string $ w $, denoted by $ C(w) $, is the length of the shortest program $w^* \in \mathbf{L_U} $ such that $ \mathbf{U}(w^*) = w $.\footnote{ Note that $ w^* $ denotes the lexicographically first $ \mathrm{p} \in \mathbf{L_U} $ such that $ l(\mathrm{p}) $ is minimum and $ \mathbf{U}(p) = w $.} The \emph{conditional} plain algorithmic complexity of a binary finite string $ y $ given a binary finite string $ x $, denoted by $ C( y \, | x ) $, is the length of the shortest program $w \in \mathbf{L_U} $ such that $ \mathbf{U}( \left< x , w \right> ) = y $. Note that $ C( y ) = C( y \, | \epsilon ) $, where $ \epsilon $ is the empty string. We also have the \emph{joint} plain algorithmic complexity of strings $x$ and $y$ denoted by $ C( x , y ) \coloneqq C( \left< x , y \right> ) $ and the \emph{C-complexity of information} in $x$ about $y$ denoted by $ I_C( x : y ) \coloneqq C(y) - C( y \, | x ) $.

	\begin{notation}\label{BdefPlaincomplexityofedgesets}
		Let $ \left( e_1 , \dots , e_{ \left| \mathbb{E}( \mathscr{G} )  \right| } \right) $ be a previously fixed ordering (or indexing) of the set $ \mathbb{E}( \mathscr{G} ) $.
		For an (composite) edge set $ \mathscr{E}( \mathscr{G} ) $, let $ C( \mathscr{E}( \mathscr{G} ) ) \coloneq C( \left< \mathscr{E}( \mathscr{G} ) \right> ) $, where $\left< \mathscr{E}( \mathscr{G} ) \right>$ denotes the (composite) edge set string
		\[
		\left<  \left< e_1, z_1 \right>, \dots , \left< e_n , z_n \right>  \right> 
		\]
		\noindent such that 
		\[
		z_i = 1  \iff e_i \in \mathscr{E}( \mathscr{G} )
		\text{ ,}
		\]
		\noindent where  $ z_i \in \{ 0 , 1 \} $ with $ 1 \leq i \leq n= | \mathbb{E}( \mathscr{G} ) | $. 
		Thus, in the simple MAG case (as in Definition~\ref{defSimplifiedMAG}) with the ordering fixed in Definition~\ref{BdefCharacteristicstringofasimpleMAG},
		we will have that $\left< \mathscr{E}( \mathscr{G}_c ) \right>$ denotes the (composite) edge set string
		\[
		\left<  \left< e_1, z_1 \right>, \dots , \left< e_n , z_n \right>  \right> 
		\]
		\noindent such that 
		\[
		z_i = 1  \iff e_i \in \mathscr{E}( \mathscr{G}_c )
		\text{ ,}
		\]
		\noindent where  $ z_i \in \{ 0 , 1 \} $ with $ 1 \leq i \leq n= | \mathbb{E}_c( \mathscr{G}_c ) | $.
		The same applies analogously to the conditional, joint, and C-complexity of information cases from Definition~\ref{BdefC}.
	\end{notation}
\end{definition}

\begin{definition}\label{BdefK}
	The (unconditional) \emph{prefix} \emph{algorithmic complexity} (also known as K-complexity, prefix Kolmogorov complexity, prefix program-size complexity or prefix Solomonoff-Komogorov-Chaitin complexity) of a finite binary string $ w $, denoted by $ K(w) $, is the length of the shortest program $w^* \in \mathbf{L'_U} $ such that $ \mathbf{U}(w^*) = w $.
	The \emph{conditional} prefix algorithmic complexity of a binary finite string $ y $ given a binary finite string $ x $, denoted by $ K( y \, | x ) $, is the length of the shortest program $w \in \mathbf{L'_U} $ such that $ \mathbf{U}( \left< x , w \right> ) = y $. Note that $ K( y ) = K( y \, | \epsilon ) $, where $ \epsilon $ is the empty string. Similarly, we have the \emph{joint} prefix algorithmic complexity of strings $x$ and $y$ denoted by $ K( x , y ) \coloneqq K( \left< x , y \right> ) $, the \emph{K-complexity of information} in $x$ about $y$ denoted by $ I_K( x : y ) \coloneqq K(y) - K( y \, | x ) $, and the \emph{mutual algorithmic information} of the two strings $x$ and $y$ denoted by $ I_A( x \, ; y ) \coloneqq K(y) - K( y \, | x^* ) $.

	\begin{notation}\label{BdefPrefixcomplexityofedgesets}
		Analogously to Notation~\ref{BdefPlaincomplexityofedgesets}, for an (composite) edge set $ \mathscr{E}( \mathscr{G} ) $, let $ K( \mathscr{E}( \mathscr{G} ) ) \coloneq K( \left< \mathscr{E}( \mathscr{G} ) \right> ) $ denote 
		\[
		K( \left<  \left< e_1, z_1 \right>, \dots , \left< e_n , z_n \right>  \right> )
		\]
		\noindent such that 
		\[
		z_i = 1  \iff e_i \in \mathscr{E}( \mathscr{G} )
		\text{ ,}
		\]
		\noindent where  $ z_i \in \{ 0 , 1 \} $ with $ 1 \leq i \leq n= | \mathbb{E}( \mathscr{G} ) | $. The same applies analogously to the simple MAG case (as in Notation~\ref{BdefPlaincomplexityofedgesets}) and the conditional, joint, K-complexity of information, and mutual cases from Definition~\ref{BdefK}.
	\end{notation}
\end{definition}

\subsection{Algorithmically random multiaspect graphs}\label{appendixBackgroundrandomMAGs}

In order to study C-randomness (i.e., plain algorithmic randomness) of simple MAGs analogously to classical graphs, first one needs to extend the concept of labeling in classical graphs to families of simple MAGs \cite{Abrahao2018darxivandreport}:

%
%

\begin{definition}\label{defLabeledfamilyofMAG}
	A family $ F_{ \mathscr{G}_c } $ of simple MAGs $ \mathscr{G}_c $ (as in Definition~\ref{defSimplifiedMAG}) is \emph{recursively labeled} \textit{iff} there are programs $ \mathrm{p'}_1 , \mathrm{p'}_2  \in \{ 0 , 1 \}^* $ such that the following hold at the same time:
	
	\begin{enumerate}[label=(\Roman{*})]
		\item\label{defLabeledfamilyofMAG1} if $ \left( a_1 , \dots , a_p \right) , \left( b_1 , \dots , b_p \right) \in \mathbb{V}\left( \mathscr{G}_c \right)$, then
		\begin{align*}
		& \mathbf{U}\left( \left< \left< a_1 , \dots , a_p \right> , \left< b_1 , \dots , b_p \right> ,  \mathrm{p'}_1   \right> \right)  = { \left( j \right) }_2
		\text{ ;}
		\end{align*}
		
		\item\label{defLabeledfamilyofMAG1.2} if $  \left( a_1 , \dots , a_p \right)  $ or $ \left( b_1 , \dots , b_p \right) $ does not belong to any $ \mathbb{V}\left( \mathscr{G}_c \right) $ with $ \mathscr{G}_c \in F_{ \mathscr{G}_c } $, then
		\begin{align*}
		& \mathbf{U}\left( \left< \left< a_1 , \dots , a_p \right> , \left< b_1 , \dots , b_p \right> ,  \mathrm{p'}_1   \right> \right)  = 0
		\text{ ;}
		\end{align*}
		
		\item\label{defLabeledfamilyofMAG2} if \[ 1 \leq j \leq \left| \mathbb{E}_c( \mathscr{G}_c )   \right| = \frac{ { \left| \mathbb{V}( \mathscr{G}_c ) \right| }^2 - { \left| \mathbb{V}( \mathscr{G}_c ) \right| } }{ 2 }  \text{ ,}\] then
		\begin{align*}
		&\mathbf{U}\left( \left<  j  , \mathrm{p'}_2  \right>  \right) = \left< \left< a_1 , \dots , a_p  \right> , \left<  b_1 , \dots , b_p \right>  \right> = { \left( e_j \right) }_2
		\text{ ;}
		\end{align*}
		
		\item\label{defLabeledfamilyofMAG.2} if \[ 1 \leq j \leq \left| \mathbb{E}_c( \mathscr{G}_c )   \right| = \frac{ { \left| \mathbb{V}( \mathscr{G}_c ) \right| }^2 - { \left| \mathbb{V}( \mathscr{G}_c ) \right| } }{ 2 }  \] does not hold for any $ \mathbb{V}\left( \mathscr{G}_c \right) $ with $ \mathscr{G}_c \in F_{ \mathscr{G}_c } $, then
		\begin{align*}
		& \mathbf{U}\left( \left<  j  , \mathrm{p'}_2  \right>  \right) = \left< \left< a_1 , \dots , a_p  \right> , \left<  b_1 , \dots , b_p \right>  \right> = \left< 0 \right>
		\text{ ;}
		\end{align*}
	\end{enumerate}
	where $ \mathscr{G}_c \in F_{ \mathscr{G}_c } $, $ a_i , b_i , j \in \mathbb{N} $, and $ 1 \leq i \leq p \in \mathbb{N}$.
	%
	
\end{definition}

Thus, in Definition~\ref{defLabeledfamilyofMAG}, note that $ \mathrm{p'}_1 $ and $ \mathrm{p'}_2 $ define a family of finite MAGs (with arbitrarily fixed order $p$) for which there is a unique ordering for all possible composite edges in the sets $ \mathbb{E}_c( \mathscr{G}_c ) $'s and this ordering does not depend on the choice of the MAG $ \mathscr{G}_c $ in the family $ F_{ \mathscr{G}_c } $.
Indeed, we will see in Lemma~\ref{lemmaLabeledfamilyofMAG} that, given an arbitrarily fixed order $p$, there is a recursively labeled infinite family that contains every possible MAG of order $p$, so that such a unique ordering of composite edges is shared by each one the member of the family.

Second, we extend the definition of plain algorithmically random classical graphs in \cite{Li1997,Buhrman1999} to simple MAGs. 
First, in the classical graph case, we have in \cite{Li1997,Buhrman1999}:

\begin{definition}\label{defC-randomgraph}
	
	A classical graph $ G $  with $ | V(G) | = n $ is $ \delta(n) $-C-random if, and only if, it satisfies
	\[
	C( E(G) \, | n ) \geq \binom{n}{2} - \delta(n)
	\text{ ,}
	\]
	\noindent where 
	\[ \myfunc{ \delta }{ \mathbb{N} }{ \mathbb{N} }{ n  }{ \delta(n) } \] 
	\noindent is the randomness deficiency function.
\end{definition} 

Then, in the simple MAG case, an analogous definition holds as introduced in \cite{Abrahao2018darxivandreport}:

\begin{definition}\label{defC-randomMAG}
	We say a simple MAG $ \mathscr{G}_c $ is $ \delta( \left| \mathbb{V}( \mathscr{G}_c ) \right| ) $-C-random \textit{iff} it satisfies
	\[
	C\left( \mathscr{E}( \mathscr{G}_c ) \, \mid \left| \mathbb{V}( \mathscr{G}_c ) \right| \right) 
	\geq 
	\binom{ \left| \mathbb{V}( \mathscr{G}_c ) \right| }{ 2 } - \delta( \left| \mathbb{V}( \mathscr{G}_c ) \right| )
	\text{ ,}
	\]
	\noindent where 
	\[ \myfunc{ \delta }{ \mathbb{N} }{ \mathbb{N} }{ n  }{ \delta(n) } \] 
	\noindent is the randomness deficiency function.
\end{definition}

With respect to weak asymptotic dominance of function $f$ by a function $g$, 
we employ the usual $\mathbf{O}( g(x) )$ for the big \textbf{O} notation when $f$ is asymptotically upper bounded by $g$; 
and with respect to strong asymptotic dominance by a function $g$, we employ the usual $\mathbf{o}( g(x) )$ when $g$ dominates $f$.

As also introduced in \cite{Abrahao2018darxivandreport}, one can also apply this same concept of algorithmic randomness to the prefix-free (or self-delimited) version, i.e., K-randomness:

\begin{definition}\label{defK-randomMAGs}
	We say a simple MAG $ \mathscr{G}_c $ is  $  \mathbf{O}(1) $-K-random \textit{iff} it satisfies
	\[
	K( \mathscr{E}( \mathscr{G}_c ) ) \geq \binom{ \left| \mathbb{V}( \mathscr{G}_c ) \right| }{ 2 } - \mathbf{O}(1)
	\]
\end{definition}

\section{Background results}\label{sectionBackground}

In this section, we briefly recover some previous results. 

\subsection{Algorithmic information theory}\label{subsectionAIT}

First of all, it is important to remember some basic and important relations in algorithmic information theory \cite{Li1997,Chaitin2004,Downey2010,Calude2009}.  

\begin{lemma}\label{lemmaBasicAIT}
	For every $ x , y \in \{ 0 , 1 \}^* $ and $  n \in \mathbb{N} $,
	\begin{align}
	\label{lemmaBasicAIT1} C( x ) & \leq l( x ) + \mathbf{O}(1) \\
	\label{lemmaBasicAIT10} K( x ) & \leq l( x ) + \mathbf{O}( \lg( l( x ) ) ) \\
	\label{lemmaBasicAIT2} C( y \, | x ) & \leq C( y ) + \mathbf{O}(1) \\
	\label{lemmaBasicAIT11} K( y \, | x ) & \leq K( y ) + \mathbf{O}(1) \\
	\label{lemmaBasicAIT3} C( y \, | \, x ) \leq K( y \, | \, x ) + \mathbf{O}(1) & \leq C( y \, | \, x ) + \mathbf{O}( \lg( C( y \, | \, x ) ) ) \\
	\label{lemmaBasicAIT4} C( x ) \leq C( x , y ) + \mathbf{O}(1) & \leq C( y ) + C( x \, | y ) + \mathbf{O}( \lg( C( x , y ) ) ) \\
	\label{lemmaBasicAIT5} K( x ) \leq K( x , y ) + \mathbf{O}(1) & \leq  \; K( y ) + K( x \, | y ) + \mathbf{O}(1) \\
	\label{lemmaBasicAIT6} C( x ) & \leq K( x ) + \mathbf{O}(1) \\
	\label{lemmaBasicAIT7} K( n ) & = \mathbf{O}( \lg( n ) ) \\
	\label{lemmaBasicAIT8} K( x ) & \leq C( x ) + K( C( x ) ) + \mathbf{O}(1) \\
	\label{lemmaBasicAIT9} I_A( x ; y ) & = I_A( y ; x ) \pm \mathbf{O}(1)
	\end{align}

\end{lemma}

Note that the inverse relation $ K( x , y ) + \mathbf{O}(1) \geq  \; K( y ) + K( x \, | y ) + \mathbf{O}(1)  $ does not hold in general in Equation~\eqref{lemmaBasicAIT5}. In fact, one can show that $ K( x , y ) = K( y ) + K( x \, \mid \, \left< y , K( y ) \right>  ) \pm \mathbf{O}(1) $, which is the key step to prove Equation~\eqref{lemmaBasicAIT9}.
In this way, we have that the notion of \emph{network topological (algorithmic) information} of a simple MAG $ \mathscr{G}_c $, i.e., the computably irreducible information necessary to determine/compute $ \left< \mathscr{E}( \mathscr{G}_c ) \right> $, is formally captured by $ I_A( \left< \mathscr{E}( \mathscr{G}_c ) \right> \, ; \left< \mathscr{E}( \mathscr{G}_c ) \right> ) = K( \left< \mathscr{E}( \mathscr{G}_c ) \right> ) \pm \mathbf{O}(1) $.
In the present article, wherever the concept of information is mentioned, we are in fact referring to algorithmic information; and wherever the term topological information appears, it is referring to the computably irreducible information necessary to determine/compute the composite edge set $ \mathscr{E}( \mathscr{G}_c ) $, which is structured and represented by string $ \left< \mathscr{E}( \mathscr{G}_c ) \right> $.

\subsection{Multiaspect graphs and graphs isomorphism}

In order to represent multidimensional networks, we are basing our work on a generalized graph representation of dyadic relations between $n$-tuples \cite{Wehmuth2016b,Wehmuth2017} called multiaspect graphs (MAGs). Directly from \cite{Wehmuth2016b}, one has that a simple MAG is basically equivalent to a classical (i.e., simple labeled) graph:

\begin{corollary}\label{corClassicMAGisomorphism}
	For every simple MAG $ \mathscr{G}_c $ of order $p>0$, where all aspects are non-empty sets, there is a unique (up to a graph isomorphism) classical graph $ G_{ \mathscr{G}_c } = \left( V , E \right) $ with $ | V( G ) | =  \prod\limits_{ n = 1 }^{ p } \left| \mathscr{A}( \mathscr{G}_c )[ n ] \right| $ that is isomorphic (as in Definition~\ref{defMAGandgraphsisomorphisms}) to $ \mathscr{G}_c $.
\end{corollary}

We also have that the concepts of \emph{walk}, \emph{trail}, and \emph{path} become well-defined for MAGs analogously to within the context of graphs. See \cite[Section 3.5]{Wehmuth2016b}.

\subsection{Algorithmically random multiaspect graphs}\label{subsectionRandomMAGs}

In this section, we briefly recover some previous results on algorithmically random graphs in \cite{Abrahao2018darxivandreport}.
These are obtained by basing our previous work on algorithmically random classical (i.e., labeled simple) graphs in \cite{Buhrman1999,Zenil2018a,Khoussainov2014}.

In order to study \emph{plain} algorithmic randomness (i.e., C-randomness) of simple MAGs, we extended in \cite{Abrahao2018darxivandreport} the concept of labeling in classical graphs to simple MAGs, as restated in Definiton~\ref{defLabeledfamilyofMAG}, while avoiding the algorithmic information distortions due to isomorphic transformations---on finite networks instead of infinite ones, and for isomorphisms in general instead automorphisms only---that were missing in previous work~\cite{Buhrman1999,Khoussainov2014}.
Then, we directly extended the definition of algorithmically random classical graphs to simple MAGs as restated in Definition~\ref{defK-randomMAGs}. 
Thus, by combining Corollary~\ref{corClassicMAGisomorphism} with well-known inequalities in algorithmic information theory \cite{Li1997,Downey2010,Chaitin2004}, we showed in \cite{Abrahao2018darxivandreport} that:

\begin{theorem}\label{thmC-randomMAGsandgraphs}
	Let $  F_{ \mathscr{G}_c } \neq \emptyset  $ be an arbitrary recursively labeled family of simple MAGs $ \mathscr{G}_c $.
	Then, for every $  \mathscr{G}_c  \in  F_{ \mathscr{G}_c } $,
	\begin{center}
		$ \mathscr{G}_c$ is $ \left( \delta( \left| \mathbb{V}( \mathscr{G}_c ) \right| ) + \mathbf{O}( \log_2( \left| \mathbb{V}( \mathscr{G}_c ) \right|) ) \right) $-C-random \\ 
		\textit{iff} \\ 
		$ G $ is $ \left( \delta( \left| V( G ) \right| ) + \mathbf{O}( \log_2( \left| V( G ) \right| ) ) \right) $-C-random, 
	\end{center}
	\noindent where $ G $ is isomorphic (as in Definition~\ref{defMAGandgraphsisomorphisms}) to $ \mathscr{G}_c $.
	
\end{theorem}

Thus, Theorem~\ref{thmC-randomMAGsandgraphs} establishes an algorithmic-informational cost of performing a transformation of a MAG into its isomorphically correspondent graph, and vice-versa.
Although such an isomorphism (as in Definition~\ref{defMAGandgraphsisomorphisms}) is an abstract mathematical equivalence, when dealing with structured data representations of objects (e.g., as the strings $ \left< \mathscr{E}( \mathscr{G}_c ) \right> $ or $ \left< E( G ) \right> $), this equivalence must always be performed by computable procedures. 
Therefore, there is always an associated necessary algorithmic information, which corresponds to the length of the minimum program that performs these transformations.
Theorem~\ref{thmC-randomMAGsandgraphs} gives a worst-case algorithmic-informational cost in a logarithmic order of the network size for plain algorithmic complexity (i.e., for non-prefix-free universal programming languages).
However, even in the prefix case (as in Definition~\ref{BdefK}), we also showed in \cite{Abrahao2018darxivandreport} that there is a algorithmic-informational cost upper bounded by a positive constant.

We also showed in \cite{Abrahao2018darxivandreport} that, given an arbitrarily fixed order $p$, there is a recursively labeled (see Definition~\ref{defLabeledfamilyofMAG}) infinite family that contains every possible MAG of order $p$:

\begin{lemma}\label{lemmaLabeledfamilyofMAG}
	There is a recursively labeled infinite family $ F_{ \mathscr{G}_c } $ of simple MAGs $ \mathscr{G}_c $ with arbitrary symmetric adjacency matrix such that every one of them has the same order $p$. 
\end{lemma}

Thus, Lemma~\ref{lemmaLabeledfamilyofMAG} ensures that one can follow the same ordering (or indexing) for the set $ \mathbb{E}_c( \mathscr{G}_c ) $, where this ordering does not depend on the choice of $ \mathscr{G}_c $, so that every possible MAG of order $p$ with vertex labels in $ \mathbb{N} $ belongs to such an infinite family satisfying Lemma~\ref{lemmaLabeledfamilyofMAG}.

In addition, the basic algorithmic-informational properties of the binary string that determines the presence or absence of a composite  edge in Definition~\ref{BdefCharacteristicstringofasimpleMAG} (i.e., the characteristic string) can be properly defined for recursively labeled families:

\begin{corollary}\label{corFamilyoflabeledMAGandstrings}
	Let $  F_{ \mathscr{G}_c }   $ be a recursively labeled family of simple MAGs $ \mathscr{G}_c $.
	Then, for every $ \mathscr{G}_c  \in F_{ \mathscr{G}_c } $ and $ x \in \{ 0 , 1 \}^* $,
	where $ x $ is the characteristic string of $  \mathscr{G}_c  $, the following relations hold:
	\begin{align}
		\label{corFamilyoflabeledMAGandstrings1} C( \mathscr{E}( \mathscr{G}_c  ) \, | \, x  ) \leq K( \mathscr{E}( \mathscr{G}_c  ) \, | \, x  ) + \mathbf{O}(1)  & =   \mathbf{O}(1)  \\
		\label{corFamilyoflabeledMAGandstrings2} C( x \, | \, \mathscr{E}( \mathscr{G}_c  )  )  \leq K( x \, | \, \mathscr{E}( \mathscr{G}_c  )  ) + \mathbf{O}(1)  & =   \mathbf{O}(1) \\
		\label{corFamilyoflabeledMAGandstrings3} K( x ) & = K( \mathscr{E}( \mathscr{G}_c  ) )  \pm \mathbf{O}\left( 1 \right) \\
		\label{corFamilyoflabeledMAGandstrings4} I_A( x ; \mathscr{E}( \mathscr{G}_c  ) ) = I_A( \mathscr{E}( \mathscr{G}_c  ) ; x ) \pm \mathbf{O}(1) & = K( x ) - \mathbf{O}\left( 1 \right) = K( \mathscr{E}( \mathscr{G}_c  ) ) \pm \mathbf{O}(1)
	\end{align}
\end{corollary}

In summary, Theorem~\ref{thmC-randomMAGsandgraphs} shows that the plain algorithmic complexity of simple MAGs and of its respective isomorphic classical graphs are roughly the same, except for the amount of algorithmic information necessary to encode the length of the program that performs this isomorphism on an arbitrary universal Turing machine. 
In fact, regarding the connections through composite edges, this shows that 
not only ``most'' of the network topological properties of such graph are inherited by the MAG (and vice-versa), but also ``most'' of those that derives from the graph's topological incompressibility.  
That is, more formally, every network topological property regarding the connections through composite edges that derives from the MAG $ \mathscr{G}_c $ being $ \left( \delta( \left| \mathbb{V}( \mathscr{G}_c ) \right| ) + \mathbf{O}( \log_2( \left| \mathbb{V}( \mathscr{G}_c ) \right|) ) \right) $-C-random is inherited by $ \mathscr{G}_c $ from its isomorphic graph $G$, if $ G $ is  $ \left( \delta( \left| V( G ) \right| ) + \mathbf{O}( \log_2( \left| V( G ) \right| ) ) \right) $-C-random and $  \left| V( G ) \right| $ is large enough. 
And the inverse inheritance also holds.
For example, we extended in \cite{Abrahao2018darxivandreport} some results from \cite{Li1997,Buhrman1999} on plain algorithmically random classical graphs to simple MAGs:

\begin{corollary}\label{corC-randomMAGs}
	Let $  F_{ \mathscr{G}_c } $ be an arbitrary recursively labeled infinite family of simple MAGs $ \mathscr{G}_c $.
	Then, the following hold for large enough $ \mathscr{G}_c  \in F_{ \mathscr{G}_c } $:
	\begin{enumerate}
		
		\item The degree $ \mathbf{d}( \mathbf{v} ) $ of a composite vertex $ \mathbf{v} \in \mathbb{V}( \mathscr{G}_c ) $ in a $ \delta( \left| \mathbb{V}( \mathscr{G}_c ) \right|  ) $-C-random MAG $  \mathscr{G}_c \in F_{ \mathscr{G}_c } $ satisfies
		\[
		\left| \mathbf{d}( \mathbf{v} ) - \left( \frac{ \left| \mathbb{V}( \mathscr{G}_c ) \right| - 1 }{ 2 } \right) \right| 
		= 
		\mathbf{O}\left( \sqrt{ \left| \mathbb{V}( \mathscr{G}_c ) \right| \, \left( \delta( \left| \mathbb{V}( \mathscr{G}_c ) \right| ) + \mathbf{O}( \log_2( \left| \mathbb{V}( \mathscr{G}_c ) \right|) )  \right) } \right)
		\]
		\label{corC-randomMAGs2}
		
		\item $ \mathbf{o}( \left| \mathbb{V}( \mathscr{G}_c ) \right|  ) $-C-random MAGs $  \mathscr{G}_c \in F_{ \mathscr{G}_c } $ have 
		\[ 
		\frac{\left| \mathbb{V}( \mathscr{G}_c ) \right| }{4} \pm \mathbf{o}(\left| \mathbb{V}( \mathscr{G}_c ) \right| )
		\]
		\noindent disjoint paths of length 2 between each pair of composite vertices $ \mathbf{u} , \mathbf{v} \in \mathbb{V}( \mathscr{G}_c ) $. 
		In particular, $ \mathbf{o}( \left| \mathbb{V}( \mathscr{G}_c ) \right|  ) $-C-random MAGs $  \mathscr{G}_c \in F_{ \mathscr{G}_c } $ have composite diameter $2$.
		\label{corC-randomMAGs3}

	\end{enumerate}
	
\end{corollary}

Thus, an incompressible simple MAG under randomness deficiency $ \delta( \left| \mathbb{V}( \mathscr{G}_c ) \right|  )  = \mathbf{o}(\left| \mathbb{V}( \mathscr{G}_c ) \right| ) $ tends to be an expected ``almost regular'' graph in the limit when the network size increases indefinitely; 
for sufficiently large set of composite vertices, these MAGs also cross a phase transition in which the diameter between composite vertices becomes $2$;
with respect to k-connectivity, as defined in \cite{Buhrman1999}, they are $ \frac{\left| \mathbb{V}( \mathscr{G}_c ) \right| }{4} \pm \mathbf{o}(\left| \mathbb{V}( \mathscr{G}_c ) \right| ) $-connected. 

Furthermore, from Definition~\ref{defK-randomMAGs},
one can build an infinite family of simple MAGs in which every member is $ \mathbf{O}( 1 ) $-K-random and, in turn, also retrieve their plain algorithmically randomness \cite{Abrahao2018darxivandreport}: 

\begin{lemma}\label{lemmaK-randomMAGs}
	There is a recursively labeled \emph{infinite} family $ F_{ \mathscr{G}_c } $ of simple MAGs $ \mathscr{G}_c $ that are $ \mathbf{O}(1) $-K-random.  
\end{lemma}

\begin{theorem}\label{thmK-randomandC-randomMAGs}
	Let $  F_{ \mathscr{G}_c }   $ be a recursively labeled infinite family of simple MAGs $ \mathscr{G}_c $ such that, for every $ \mathscr{G}_c  \in F_{ \mathscr{G}_c }  $ and $ n \in \mathbb{N} $, if $ x \upharpoonright_n $ is its characteristic string and $ n = \left| \mathbb{E}_c( \mathscr{G}_c )   \right| $,
	then  $ x \in \left[ 0 , 1 \right] \subset \mathbb{R} $ is $ \mathbf{O}(1) $-K-random.  
	Thus, every  MAG $ \mathscr{G}_c  \in F_{ \mathscr{G}_c } $ is $ \mathbf{O}( \log_2( \left| \mathbb{V}( \mathscr{G}_c ) \right| ) ) $-C-random and $ \mathbf{O}(1) $-K-random.
	In addition, there is such a family $  F_{ \mathscr{G}_c }   $ with $ \Omega = x \in \left[ 0 , 1 \right] \subset \mathbb{R}  $, where $ \Omega $ is the halting probability \cite{Chaitin2004,Calude2002}.
\end{theorem}

\section{Proof of Theorem 2.1}\label{sectionProofthmSimplespatialTVG}

\begin{theorem}[Theorem~\ref{thmSimplespatialTVG}]\label{SIthmSimplespatialTVG}
	Let $ \mathrm{G'_t} = (\mathrm{V},\mathscr{E},\mathrm{T}) $ be a simple spatial TVG that belongs to a recursively labeled infinite family $ F_{ G'_t } $ of simple TVGs.
	Then, there is a binary string $ y  \in \{ 0 , 1 \}^*  $ that is an \emph{algorithmically characteristic string} of $ \mathrm{G'_t} $ such that
	\[
	K( y ) 
	\leq
	l(y) 
	+ \mathbf{O}(1)
	\leq 
	\left| \mathrm{T}( \mathrm{G'_t} ) \right| \left( \frac{  \left| \mathrm{V}( \mathrm{G'_t} ) \right|^2 - \left| \mathrm{V}( \mathrm{G'_t} ) \right| }{ 2 } \right)  
	+ 
	\mathbf{O}\left( \log_2\left( \left| \mathrm{V}( \mathrm{G'_t} ) \right| \, \left| \mathrm{T}( \mathrm{G'_t} ) \right| \right) \right)
	\text{ ,}
	\]
	\[
	K( x ) 
	\leq 
	K(y) + \mathbf{O}( 1 )
	\text{ ,}
	\]
	\[
	K( x  \mid {y}^* ) \leq K( x  \mid {y} ) + \mathbf{O}\left( 1 \right) \leq \mathbf{O}\left( 1 \right)
	\text{ ,}
	\]
	\[
	K( y ) \leq K(x) + \mathbf{O}( 1 )
	\text{ ,}
	\]
	and
	\[
	K( y \mid x^* ) \leq K( y \mid x ) + \mathbf{O}( 1 ) \leq \mathbf{O}( 1 )
	\]
	hold, where $x$ is the characteristic string of $ \mathrm{G'_t} $.
	
	\begin{proof}
		The main idea of the proof is based on showing a Turing equivalence between the string $y$ and its respective characteristic string $x$.
		Then, as in the proof of Corollary~\ref{corFamilyoflabeledMAGandstrings} presented in \cite{Abrahao2018darxivandreport}, we will recover the Turing equivalence between the characteristic string $ x $ and the string $ \left< \mathscr{E}( \mathrm{G'_t} ) \right>  $.
		So, first let $ y' \in \{ 0 , 1 \}^*  $ be any arbitrary binary string with
		\[
		l( y' )  =  \left| \mathrm{T}( \mathrm{G'_t} ) \right| \left( \frac{  \left| \mathrm{V}( \mathrm{G'_t} ) \right|^2 - \left| \mathrm{V}( \mathrm{G'_t} ) \right| }{ 2 } \right)  
		\]
		\noindent Let $ \mathrm{p'}_2 \in \{ 0 , 1 \}^* $ be a binary string that represents an algorithm running on a prefix universal Turing machine $ \mathbf{U} $ that takes $ j \in \mathbb{N} $ as input and returns the $ j\text{-th} $ edge in $ \mathbb{E}_c( \mathrm{G'_t} ) $. The existence of such $ \mathrm{p'}_2 $ is guaranteed by the definition of recursively labeling in \cite{Abrahao2018darxivandreport} (see Definition~\ref{defLabeledfamilyofMAG}). Moreover, $ \mathrm{p'}_2 $ is independent of the choice of $ \mathrm{G'_t} $ in the family $  F_{ \mathrm{G'_t} }   $. 
		Let $ s_1 \in \{ 0 , 1 \}^* $ be a binary string that represents an algorithm running on a prefix universal Turing machine $ \mathbf{U} $ that:
		\begin{enumerate}
			\item takes $ \mathrm{p'}_2 $ , $ y' $ and $ j  $ as inputs; 
			\item calculates $ \mathbf{U}( \left< j' , \mathrm{p'}_2 \right> )  $ for every $ j' \leq j $; 
			\item enumerates all the spatial edges $ e_{j'} =  \mathbf{U}( \left< j' , \mathrm{p'}_2 \right> )  $ as a subsequence of the sequence of possible undirected edges in $ \left( e_1 , \dots , e_j \right) $; 
			\item and returns:\footnote{ Alternatively, in the case the sequential couplings were not excluded in the representation of the snapshot-dynamic network, one can add here a clause also returning $1$ if $ \mathbf{U}( \left< j , \mathrm{p'}_2 \right> )  $ is a sequential coupling. The reader is invited to note that the theorem holds anyway.} 
			\begin{enumerate}
				\item $0$, if $ \mathbf{U}( \left< j , \mathrm{p'}_2 \right> )  $ is not a spatial edge; 
				\item $1$, if $ \mathbf{U}( \left< j , \mathrm{p'}_2 \right> )  $ is the $i$-th spatial edge and the $i$-th digit in $ y' $ is $1$; 
				\item $0$, if $ \mathbf{U}( \left< j , \mathrm{p'}_2 \right> )  $ is the $i$-th spatial edge and the $i$-th digit in $ y' $ is $0$.
			\end{enumerate}
		\end{enumerate} 
		Note that deciding whether an edge $ e $ is spatial or not follows directly from deciding whether $ t_u = t_v $ or not in $ e = ( u , t_u , v , t_v ) $, which is a decidable (and computationally cheap) procedure.
		Let $ s_2 \in \{ 0 , 1 \}^* $ be a binary string that represents an algorithm running on a prefix universal Turing machine  $ \mathbf{U} $ that: 
		
		\begin{enumerate}
			\item takes $s_1$, $ \mathrm{p'}_2 $ , $ y' $ and 
			\[
			\left( \frac{  \left( \left| \mathrm{V}( \mathrm{G'_t} ) \right| \, \left| \mathrm{T}( \mathrm{G'_t} ) \right| \right)^2 - \left| \mathrm{V}( \mathrm{G'_t} ) \right| \, \left| \mathrm{T}( \mathrm{G'_t} ) \right| }{ 2 } \right)
			\]
			\noindent as inputs; 
			\item calculates $ \mathbf{U}( \left< j , y' , \mathrm{p'}_2 , s_1 \right> )   $ for every $j$ with 
			\[ 1 \leq  j \leq \left( \frac{  \left( \left| \mathrm{V}( \mathrm{G'_t} ) \right| \, \left| \mathrm{T}( \mathrm{G'_t} ) \right| \right)^2 - \left| \mathrm{V}( \mathrm{G'_t} ) \right| \, \left| \mathrm{T}( \mathrm{G'_t} ) \right| }{ 2 } \right) \text{ ;}\] 
			\noindent 
			\item and returns the binary string
			\[
			x =  z_1 \, \dots \, z_n  
			\]
			\noindent such that
			\[
			z_i = 1  \iff \mathbf{U}( \left< i , y' , \mathrm{p'}_2 , s_1 \right> )  = 1
			\text{ ,}
			\]
			\noindent where  $ z_i \in \{ 0 , 1 \} $ with $ 1 \leq i \leq n = \left( \frac{  \left( \left| \mathrm{V}( \mathrm{G'_t} ) \right| \, \left| \mathrm{T}( \mathrm{G'_t} ) \right| \right)^2 - \left| \mathrm{V}( \mathrm{G'_t} ) \right| \, \left| \mathrm{T}( \mathrm{G'_t} ) \right| }{ 2 } \right) $.
		\end{enumerate}
		Now, let $  y = \left< k , y' , \mathrm{p'}_2 , s_1 , s_2 \right> $,  where $k$ is  the self-delimiting binary representation of 
		\[  \left( \frac{  \left( \left| \mathrm{V}( \mathrm{G'_t} ) \right| \, \left| \mathrm{T}( \mathrm{G'_t} ) \right| \right)^2 - \left| \mathrm{V}( \mathrm{G'_t} ) \right| \, \left| \mathrm{T}( \mathrm{G'_t} ) \right| }{ 2 } \right) \in \mathbb{N}  \text{ .}\]
		We know one can encode $k$ in $ \mathbf{O}\left( \log_2\left(  \left( \left| \mathrm{V}( \mathrm{G'_t} ) \right| \, \left| \mathrm{T}( \mathrm{G'_t} ) \right| \right)^2  \right) \right) $ bits.
		Therefore, since $ \mathrm{p'}_2 $, $ s_1 $ , and $ s_2 $ are fixed and independent of the choice of $ \mathrm{G'_t} $, we will have that, by the minimality of the prefix algorithmic complexity,
		\begin{equation*}
		\begin{aligned}
			K( x ) 
			\leq 
			K(y) + \mathbf{O}( 1 )
			& \leq 
			l( \left< k , y' , \mathrm{p'}_2 , s_1 , s_2 \right>  )  
			+ 
			\mathbf{O}( 1 ) 
			\leq \\
			& \leq 
			l( y' ) + \mathbf{O}\left( \log_2\left(  \left( \left| \mathrm{V}( \mathrm{G'_t} ) \right| \, \left| \mathrm{T}( \mathrm{G'_t} ) \right| \right)^2  \right) \right)
		\end{aligned}
		\end{equation*}
		\noindent and
		\[
		K( x  \mid {y}^* ) \leq K( x  \mid {y} ) + \mathbf{O}\left( 1 \right) \leq \mathbf{O}\left( 1 \right)
		\text{ .}
		\]
		On the other hand, in order to show that $ K( y ) \leq K( x ) + \mathbf{O}( 1 ) $,
		\noindent let $ s_3 \in \{ 0 , 1 \}^* $ be a binary string that represents an algorithm running on a prefix universal Turing machine $ \mathbf{U} $ that:
		
		\begin{enumerate}
			\item takes $ \mathrm{p'}_2 $, $ s_1 $ , $ s_2 $, and $x$ as inputs; 
			\item enumerates all the spatial edges using program $ \mathrm{p'}_2 $; 
			\item and builds the binary string $ y' = y'_1 \, \dots \, y'_{k'} $ 
			such that: 
			\begin{enumerate}
				\item $ y'_i = 1 $, if $j$ corresponds to the $i$-th spatial edge and the $j$-th digit of x is $1$;  
				\item $ y'_i = 0 $, if $j$ corresponds to the $i$-th spatial edge and the $j$-th digit of x is $0$;
			\end{enumerate}
			\item finally, $ s_3 $ returns the binary string $ \left< k , y' , \mathrm{p'}_2 , s_1 , s_2 \right>  = y $. 
		\end{enumerate}
		Note that $ x $ was already given as input and
		\[
		l(x) 
		= 
		\left( \frac{  \left( \left| \mathrm{V}( \mathrm{G'_t} ) \right| \, \left| \mathrm{T}( \mathrm{G'_t} ) \right| \right)^2 - \left| \mathrm{V}( \mathrm{G'_t} ) \right| \, \left| \mathrm{T}( \mathrm{G'_t} ) \right| }{ 2 } \right)  \text{ .}
		\] 
		Therefore, since $ \mathrm{p'}_2 $, $ s_1 $ , $ s_2 $, and $ s_3 $ are fixed and independent of the choice of $ \mathrm{G'_t} $, we will have that, by the minimality of the prefix algorithmic complexity,		
		\[
		K( y )  = K( \left< k , y' , \mathrm{p'}_2 , s_1 , s_2 \right>  )\leq K( \left< x , \mathrm{p'}_2 , s_1 , s_2, s_3 \right>  ) + \mathbf{O}( 1 )  \leq K(x) + \mathbf{O}( 1 )
		\]
		\noindent and
		\[
		K( y \mid x^* ) \leq K( y \mid x ) + \mathbf{O}( 1 ) \leq \mathbf{O}( 1 )
		\]
		Now, 
		let $ p $ be a binary string that represents the algorithm running on a universal Turing machine that: 
		\begin{enumerate}
			\item receives the string $x$ as its input;
			\item for $ 1 \leq j \leq l(x) $, reads each $j$-th bit of $x$;
			\item calculates $   \mathbf{U}\left( \left<  j  , \mathrm{p'}_2  \right> \right)  $;
			\item and, from the outputs $ e_j $ of these programs $   \left<  j  , \mathrm{p'}_2  \right>  $, returns the string $ \left<  \left< e_1, z_1 \right>, \dots , \left< e_n , z_n \right>  \right> $, where: $ z_j = 1 $, if the $j$-th bit of $x$ is $1$; and $ z_j = 0 $, if the $j$-th bit of $x$ is $0$.
		\end{enumerate} 
		Thus, since $ \mathrm{p'}_2 $ is fixed, we will have that there is a self-delimiting binary encoding of $   p  $ such that
		\[
		\mathbf{U}\left(  \left< x , p \right> \right) = \left<  \left< e_1, z_1 \right>, \dots , \left< e_n , z_n \right>  \right>
		\]
		holds. 
		\noindent Then, by the minimality of $ K( \cdot ) $, we will have that
		\[
		K( \mathscr{E}( \mathscr{G}_c  ) \, | \, x  ) \leq  	l\left(  p  \right) \leq  \mathbf{O}(1)
		\]
		Analogously to program $p$, using program $ \mathrm{p'}_1 $ instead of $ \mathrm{p'}_2 $ in order to build the string $x$ from $ \left< \mathscr{E}( \mathscr{G}_c  ) \right>  $, we will have another program $q$ such that  
		there is a self-delimiting binary encoding of $   q  $ such that
		\[
		\mathbf{U}\left(   \left< \left<  \left< e_1, z_1 \right>, \dots , \left< e_n , z_n \right>  \right> , q \right>  \right) = x
		\]
		holds and, by the minimality of $ K( \cdot ) $, we have that
		\[
		K( x  \, | \, \mathscr{E}( \mathscr{G}_c  ) ) \leq  	l\left(  q  \right) \leq  \mathbf{O}(1)
		\]
		Note that $ p , q , \mathrm{p'}_1 , \mathrm{p'}_2 , s_1 , s_2, s_3 $ are fixed.
		Therefore, in order to finish the proof, just let $ \mathrm{p''}_1 $ be the binary string that represents the algorithm running on a universal Turing machine $ \mathbf{U} $ that receives $y$ as input and returns the value of
		\[ \mathbf{U}( \left< \mathbf{U}( y ) , p \right> )  = \left<  \left< e_1, z_1 \right>, \dots , \left< e_n , z_n \right>  \right> = \left< \mathscr{E}( \mathscr{G}_c  ) \right> \text{ .}\]
		Similarly,  let $ \mathrm{p''}_2 $ be the binary string that represents the algorithm running on a universal Turing machine $ \mathbf{U} $ that receives 
		$ \left< \mathscr{E}( \mathscr{G}_c  ) \right> $
		as input and returns the value of
		\[ \mathbf{U}(  \left< \mathbf{U}\left(   \left< \left< \mathscr{E}( \mathscr{G}_c  ) \right> , q \right>  \right) , \mathrm{p'}_2 , s_1 , s_2, s_3 \right> )  = y \text{ .}\]
		
%
	\end{proof}
\end{theorem}

\section{Multidimensional degree, connectivity, diameter, and non-sequential interdimensional edges}\label{sectionTopologicalTVG}

As we will show in this section, although an incompressible general multidimensional network can carry much more information in its topology than any algorithmically snapshot-like multidimensional network---see Sections~\ref{subsectionSnapshotDynamic}, \ref{subsectionMultiplex}, and \ref{sectionAlgorithmicsnapshot}---, it displays some properties that may not be seen in real-world multidimensional networks, e.g., in those that can be univocally represented by algorithmically snapshot-like multidimensional networks.

First, since a TVG is just a second order MAG, it is immediate to show in Corollary~\ref{corC-randomTVGtopologicalproperties} that the previously studied Corollary~\ref{corC-randomMAGs}, which holds for arbitrary order $p$ \cite{Abrahao2018darxivandreport}, also applies to simple TVGs.
Thus, for the sake of exemplification, we start with the dynamic case.
Then, we generalize to the multidimensional case.

\begin{corollary}\label{corC-randomTVGtopologicalproperties}
	Let $ F_{ \mathrm{ G_t } }  $ be a recursively labeled infinite family $  F_{ \mathrm{ G_t } } $ of simple TVGs $ \mathrm{ G_t } $ that are $ \mathbf{O}( \log_2( \left| \mathbb{V}( \mathrm{ G_t } ) \right| ) ) $-C-random.
	Then, the following hold for large enough $ \mathrm{ G_t } \in F_{ \mathrm{ G_t } } $, where $ \mathbb{V}( \mathrm{ G_t } ) = \mathrm{V}( \mathrm{ G_t } ) \times \mathrm{T}( \mathrm{ G_t } )  $:
	\begin{enumerate}
		\item The degree $ \mathbf{d}( \mathbf{v} ) $ of a composite vertex $ \mathbf{v} \in \mathbb{V}( \mathrm{ G_t } ) $ in a MAG $  \mathrm{ G_t } \in F_{ \mathrm{ G_t } } $ satisfies
		\[
		\left| \mathbf{d}( \mathbf{v} ) - \left( \frac{ \left| \mathbb{V}( \mathrm{ G_t } ) \right| - 1 }{ 2 } \right) \right| 
		= 
		\mathbf{O}\left( \sqrt{ \left| \mathbb{V}( \mathrm{ G_t } ) \right| \, \left(  \mathbf{O}( \log_2( \left| \mathbb{V}( \mathrm{ G_t } ) \right|) ) \right) } \right) \text{ .}
		\]
		\label{corK-randomMAGsproperties1}
		
		\item $  \mathrm{ G_t }  $ has \[ 
		\frac{\left| \mathbb{V}( \mathrm{ G_t } ) \right| }{4} \pm \mathbf{o}(\left| \mathbb{V}( \mathrm{ G_t } ) \right| )
		\] disjoint paths of length 2 between each pair of composite vertices $ \mathbf{u} , \mathbf{v} \in \mathbb{V}( \mathrm{ G_t } ) $. 
		\label{corK-randomMAGsproperties2}
		
		\item $  \mathrm{ G_t } $ has (composite) diameter $2$.
		\label{corK-randomMAGsproperties3}

	\end{enumerate}
	
\end{corollary}

From Lemma~\ref{lemmaK-randomMAGs}, it is also immediate that there is an incompressible TVG that satisfies the conditions of Corollary~\ref{corC-randomMAGs}. 
In fact, this follows from Theorem~\ref{thmK-randomandC-randomMAGs} by assuming a family of initial segments of a K-random (i.e., prefix algorithmically random) real number. 
For example, one may take initial segments of the \emph{halting probability} (or Chaitin's constant) \cite{Abrahao2018darxivandreport}. 
Thus, from Lemma~\ref{lemmaK-randomMAGs}, Corollary~\ref{corFamilyoflabeledMAGandstrings}, and Theorem~\ref{thmK-randomandC-randomMAGs}, we have that Corollary~\ref{corC-randomMAGs} is satisfiable with $ \delta( \left| \mathbb{V}( \mathscr{G}_c ) \right|  )   = \mathbf{O}( \log_2( \left| \mathbb{V}( \mathrm{ G_t } ) \right| ) ) $ and, therefore, resulting in the satisfiability of Corollary~\ref{corC-randomTVGtopologicalproperties} .

%
The short composite diameter and high $k$-connectivity of a $ \mathbf{O}( \log_2( \left| \mathbb{V}( \mathrm{ G_t } ) \right| ) ) $-C-random simple TVG ensures the existence of transtemporal edges in $ \mathrm{ G_t } $:

\begin{corollary}\label{corTranstemporaledges}
	Let $ \mathrm{ G_t } $ be any large enough simple TVG satisfying Corollary~\ref{corC-randomTVGtopologicalproperties} with 
	\[ 
	\left(
	\frac{1}{ 16 } - \frac{ \mathbf{o}( \left| \mathrm{V}( \mathrm{ G_t } ) \times \mathrm{T}( \mathrm{ G_t } )  \right| ) }{ 4 \, \left| \mathrm{V}( \mathrm{ G_t } ) \times \mathrm{T}( \mathrm{ G_t } )  \right| }
	\right)^{-1}
	=
	\mathbf{o}
	\left( 	
	\left| \mathrm{ T }( \mathrm{ G_t }  ) \right|
	\right) 
	\text{ .}
	\]
	Then, between every pair of vertices $ u , v \in \mathrm{ V }( \mathrm{ G_t }  ) $ and time instants $ t_i , t_j \in \mathrm{ T }( \mathrm{ G_t }  ) $ with $ j > i+2 $, there is at least one transtemporal edge $ e \in \mathscr{ E }( \mathrm{ G_t } ) $.
	
	\begin{proof}
		A particular case of Theorem~\ref{thmTransaspectedges}.
	\end{proof}
	
\end{corollary}

In fact, Corollary~\ref{corTranstemporaledges} holds as a particular case of undirected high-order networks (i.e., undirected node-aligned general multidimensional networks), as we will demonstrate below. 
The multilayered case with just one additional aspect besides the set of vertices is totally analogous to Corollary~\ref{corTranstemporaledges}. 
For the multidimensional case, we will have that the first aspect still is the set of vertices. 
The second aspect in turn may be the set $ \mathrm{ T }( \mathscr{ G }_c ) =  \mathscr{A}( \mathscr{ G }_c  )[2]  $ of time instants or it may be the first layer type $ \mathrm{ L_1 }( \mathscr{ G }_c ) = \mathscr{A}( \mathscr{ G }_c  )[2] $). 
The further aspects are any other layer type $ \mathrm{ L_k }( \mathscr{ G }_c ) = \mathscr{A}( \mathscr{ G }_c  )[k+1] $, where $ k \leq \left| \mathscr{A}( \mathscr{ G }_c  ) \right| - 1 $, or any other node dimension.
Note that, unlike in \cite{Kivela2014}, we refer to each element of $ \mathrm{ L_k }( \mathscr{ G }_c ) $ as a \emph{layer} and to each $ \left(\alpha_1 , \dots , \alpha_k\right) \in \mathrm{ L_1 }( \mathscr{ G }_c ) \times \cdots \times \mathrm{ L_k }( \mathscr{ G }_c ) $ as a \emph{layer tuple} (or composite layer), instead of, respectively, a elementary layer and a layer.
Moreover, we refer to each set $ \mathrm{ L_k }( \mathscr{ G }_c ) $ as a \emph{layer type} and to each arbitrary set $ \mathrm{ L_i }( \mathscr{ G }_c ) \times \cdots \times \mathrm{ L_j }( \mathscr{ G }_c ) $ as a \emph{multilayer type}.


By noting that Corollary~\ref{corC-randomMAGs} applies to simple MAGs with order $ p \geq 2 $ (as proved in \cite{Abrahao2018darxivandreport}), Corollary~\ref{corTranstemporaledges} becomes indeed a particular case of:

\begin{theorem}\label{thmTransaspectedges}
	Let $ \mathscr{ G }_c $ be any large enough $ \mathbf{O}( \log_2( \left| \mathbb{V}( \mathscr{ G }_c ) \right| ) ) $-C-random simple MAG with order $ p \geq 2 $ that
	satisfies Corollary~ \ref{corC-randomMAGs} with $ \delta\left( \left| \mathbb{V}( \mathscr{ G }_c ) \right| \right) = \mathbf{O}( \log_2( \left| \mathbb{V}( \mathscr{ G }_c ) \right| ) ) $
	such that 
	\[
	\left(
	\frac{ 1 }{ 16 } - \frac{ \mathbf{o}( \left| \mathbb{ V }( \mathscr{ G }_c  ) \right| ) }{ 4 \, \left| \mathbb{ V }( \mathscr{ G }_c  ) \right| }
	\right)^{-1}
	=
	\mathbf{o}
	\left( 	
	\frac{\left| \mathbb{ V }( \mathscr{ G }_c  ) \right| }{ \left| \mathrm{ V }( \mathscr{ G }_c ) \right| \, 
		\bigtimes\limits_{ 2 \leq h \leq p , h \neq k \leq p} \left| \mathscr{A}( \mathscr{ G }_c  )[h]  \right|
	}
	\right) 
	\text{,} 
	\]
	where $ 2 \leq k \leq p $.
	Then, between every pair of composite vertices 
	\[ ( u ,  \dots  , x_{ki} , \dots ,  x_{ps} ) \] 
	and 
	\[ ( v , \dots , x_{kj} \dots , x_{ps'}) \] 
	in $  \mathbb{ V }( \mathscr{ G }_c  ) $  with $ j > i+2 $, there is: at least one crosslayer edge $ e \in \mathscr{ E }( \mathscr{ G }_c ) $, if $ \mathrm{ L_{ k - 1 } }( \mathscr{ G }_c )  = \mathscr{A}( \mathscr{ G }_c  )[k] $; at least one transtemporal edge $ e \in \mathscr{ E }( \mathscr{ G }_c ) $, if $ \mathrm{ T }( \mathscr{ G }_c )   = \mathscr{A}( \mathscr{ G }_c  )[k] $; or at least one non-sequential interdimensional edge $ e \in \mathscr{ E }( \mathscr{ G }_c ) $, if $ \mathscr{A}( \mathscr{ G }_c  )[k] $ corresponds to any arbitrary node dimension.
	
\begin{proof}[Proof]
	Given a proper interpretation of the $k$-th aspect, both transtemporal edges and crosslayer edges are particular cases of non-sequential interdimensional edges.  
	Thus, we will prove only the general case.
	First, if 
	\[ ( u ,  \dots  , x_{ki} , \dots ,  x_{ps} , v , \dots , x_{kj} \dots , x_{ps'} ) \in \mathscr{ E }(  \mathscr{ G }_c  ) \text{ ,} \]
	then it immediately satisfies the definition of non-sequential interdimensional edge. 
	Thus, it suffices to investigate the cases in which there is a $ ( v' ,  \dots  , x_{kz} , \dots ,  x_{ps''} ) \in \mathbb{ V }( \mathscr{ G }_c  ) $ with $ x_{kz} \in \mathscr{A}( \mathscr{ G }_c  )[k]  $ such that 
	\[ ( u ,  \dots  , x_{ki} , \dots ,  x_{ps} , v' ,  \dots  , x_{kz} , \dots ,  x_{ps''} )  \in \mathscr{ E }(  \mathscr{ G }_c  )   \]
	or
	\[ ( v' ,  \dots  , x_{kz} , \dots ,  x_{ps''} , v , \dots , x_{kj} \dots , x_{ps'} ) \in \mathscr{ E }(  \mathscr{ G }_c  ) \text{ .} \]
	From Corollary~\ref{corC-randomMAGs}, we have that, for every pair of composite vertices \\ $ ( u ,  \dots  , x_{ki} , \dots ,  x_{ps}  ) $ and $ ( v , \dots , x_{kj} \dots , x_{ps'} ) $, there are $ \frac{\left| \mathbb{ V }( \mathscr{ G }_c  ) \right| }{4} \pm \mathbf{o}(\left| \mathbb{ V }( \mathscr{ G }_c  ) \right| ) $ disjoint paths of length 2 between $ ( u ,  \dots  , x_{ki} , \dots ,  x_{ps}  ) $ and $ ( v , \dots , x_{kj} \dots , x_{ps'} ) $.
	Now, notice that 
	\[
	\left| \mathscr{A}( \mathscr{ G }_c  )[k]  \right|
	=
	\frac{\left| \mathbb{ V }( \mathscr{ G }_c  ) \right| }{ \left| \mathrm{ V }( \mathscr{ G }_c ) \right| \, 
		\bigtimes\limits_{ 2 \leq h \leq p , h \neq k \leq p} \left| \mathscr{A}( \mathscr{ G }_c  )[h]  \right|
	} 
	\text{ .}
	\]
	But, since $ \mathscr{ G }_c $ can have arbitrarily large sets $ \mathbb{ V }( \mathscr{ G }_c  ) $ and
	\[
	\left(
	\frac{ 1 }{ 16 } - \frac{ \mathbf{o}( \left| \mathbb{ V }( \mathscr{ G }_c  ) \right| ) }{ 4 \, \left| \mathbb{ V }( \mathscr{ G }_c  ) \right| }
	\right)^{-1}
	=
	\mathbf{o}
	\left( 	
	\frac{\left| \mathbb{ V }( \mathscr{ G }_c  ) \right| }{ \left| \mathrm{ V }( \mathscr{ G }_c ) \right| \, 
		\bigtimes\limits_{ 2 \leq h \leq p , h \neq k \leq p} \left| \mathscr{A}( \mathscr{ G }_c  )[h]  \right|
	}
	\right) 
	\text{,} 
	\]
	then the number of possible distinct composite vertices $ ( v' ,  \dots  , x_{kz} , \dots ,  x_{ps''} ) $ with $ z = i \pm 1 $ or $ z = j \pm 1 $ will be always smaller than or equal to
	\begin{equation}\label{eqTranstemporalMain}
		\begin{aligned}
			& \lim\limits_{ \left| \mathbb{ V }( \mathscr{ G }_c  ) \right| \to \infty } 
			4 \, \left( \left| \mathrm{ V }( \mathscr{ G }_c ) \right| \, 
			\bigtimes\limits_{ 4 \leq h \leq p , h \neq k \leq p} \left| \mathscr{A}( \mathscr{ G }_c  )[h]  \right| \right) = \\
			= \; & \lim\limits_{ \left| \mathbb{ V }( \mathscr{ G }_c  ) \right| \to \infty } 
			4 \, 
			\left| \mathbb{ V }( \mathscr{ G }_c  ) \right| \, 
			\left( \frac{ \left| \mathbb{ V }( \mathscr{ G }_c  ) \right| }{ \left| \mathrm{ V }( \mathscr{ G }_c ) \right| \, \bigtimes\limits_{ 2 \leq h \leq p , h \neq k \leq p} \left| \mathscr{A}( \mathscr{ G }_c  )[h]  \right|  } \right)^{ -1 }
			<< \\
			<< \; & 
			\lim\limits_{ \left| \mathbb{ V }( \mathscr{ G }_c  ) \right| \to \infty } 
			4 \, 
			\left| \mathbb{ V }( \mathscr{ G }_c  ) \right|\,
			\left(
			\frac{ 1 }{ 16 } - \frac{ \mathbf{o}( \left| \mathbb{ V }( \mathscr{ G }_c  ) \right| ) }{ 4 \, \left| \mathbb{ V }( \mathscr{ G }_c  ) \right| }
			\right) = \\
			= \; & \,
			\lim\limits_{ \left| \mathbb{ V }( \mathscr{ G }_c  ) \right| \to \infty }
			\frac{ \left| \mathbb{ V }( \mathscr{ G }_c  ) \right| }{4}
			-
			\mathbf{o}( \left| \mathbb{ V }( \mathscr{ G }_c  ) \right| )
		\end{aligned}
	\end{equation}
	From Equation~\eqref{eqTranstemporalMain}, one will have it that the number of sequential interdimensional edges will eventually become strictly smaller than the number of distinct composite vertices connecting $ ( u ,  \dots  , x_{ki} , \dots ,  x_{ps}  ) $ and $ ( v , \dots , x_{kj} \dots , x_{ps'} ) $. 
	Therefore, 
	there will be at least one composite vertex $ ( v' ,  \dots  , x_{kz} , \dots ,  x_{ps''} ) $ 
	connected by a non-sequential interdimensional edge.
\end{proof}

\end{theorem}

\end{document}